\documentclass[jmp,%
 amsmath,amssymb,
preprint,11pt]{revtex4-1}

\usepackage{graphicx}
\usepackage{dcolumn}
\usepackage{bm}
\usepackage[hyperindex]{hyperref}
\usepackage{amsmath, amssymb, amsfonts, amsthm, amscd}
\usepackage{mathrsfs}

\newtheorem{theo}{Theorem}[section]
\newtheorem{prop}[theo]{Proposition}
\newtheorem{lemma}[theo]{Lemma}
\newtheorem{cor}[theo]{Corollary}
\newtheorem{rmk}[theo]{Remark}

\newcommand{\R} {\mathbb{R}}
\newcommand{\C} {\mathbb{C}}

\newcommand{\be}{\begin{equation}}

\newcommand{\ee}{\end{equation}}
\begin{document}



\title[Asymptotic stability for standing waves of a NLS equation with concentrated nonlinearity in dimension three. II]{Asymptotic stability for standing waves of a NLS equation with concentrated nonlinearity in dimension three. II}

\author{Riccardo Adami}\email{riccardo.adami@polito.it}\affiliation{Dipartimento di Scienze Matematiche, Politecnico di Torino, C.so Duca degli Abruzzi 24 10129 Torino, Italy}
\author{Diego Noja}\email{diego.noja@unimib.it}\affiliation{Dipartimento di Matematica e Applicazioni, Universit\`a di Milano Bicocca, Via R. Cozzi 55 20125 Milano, Italy}%
\author{Cecilia Ortoleva}\email{cecilia.ortoleva@gmail.com}\affiliation{PriceWaterhouseCoopers Italia, via Monte Rosa 91, 21049, Milano }

\date{\today}
\begin{abstract}In this paper the study of asymptotic stability of
standing waves for a model of Schr\"odinger equation with spatially concentrated
nonlinearity in dimension three, begun in \cite{ADO}, is continued.
The nonlinearity studied is a power nonlinearity concentrated at the point $x=0$ obtained considering a contact (or $\delta$) interaction with strength
$\alpha$, which consists of a singular perturbation of the Laplacian
described by a selfadjoint operator $H_{\alpha}$, and letting the strength $\alpha$ depend on the wavefunction in a prescribed way: $i\dot u= H_\alpha
u$, $\alpha=\alpha(u)$. For power nonlinearities in the range
$(\frac{1}{\sqrt 2},1)$ there exist orbitally stable standing waves
$\Phi_\omega$, and the linearization around them admits two imaginary
eigenvalues (absent in the range $(0,\frac{1}{\sqrt 2})$ previously
treated) which in principle could correspond to non decaying states,
so preventing asymptotic relaxation towards an equilibrium orbit. This
situation is usually treated requiring the validity of a {\it
  nonlinear Fermi golden rule}, which assures the presence of a
dissipative term in the modulation equation ruling the complex
amplitudes $z(t), \bar z(t)$ associated to the discrete part of the
linearized spectrum. Here without the use of FGR it is proven that, in
the range $(\frac{1}{\sqrt 2},\sigma^*)$ for a certain $\sigma^* \in
(\frac{1}{\sqrt{2}}, \frac{\sqrt{3} +1}{2 \sqrt{2}}]$,  the dynamics
  near the orbit of a standing wave asymptotically relaxes in the
  following sense: consider an initial datum $u(0)$ near the standing
  wave $\Phi_{\omega_0}, $ written in the form  
$u(0) = u_0 = e^{i\omega_0 +\gamma_0} \Phi_{\omega_0} +e^{i\omega_0
+\gamma_0} [(z_0 +\overline{z_0}) \Psi_1 +i (z_0 -\overline{z_0})
\Psi_2] +f_0 $
with $z_0$ small and $f_0$ small in energy and in a certain weighted space $L^1_{w}$; then the solution $u(t)$ can be asymptotically decomposed as
$$u(t) = e^{i\omega_{\infty} t +i b_1 \log (1 +\epsilon k_{\infty} t)}
\Phi_{\omega_{\infty}} +U_t*\psi_{\infty} +r_{\infty}, \quad
\textrm{as} \;\; t \rightarrow +\infty,$$

\noindent where $\omega_{\infty}$, $k_{\infty} > 0$, $b_1
\in \R$, and $\psi_{\infty}$ and $r_{\infty} \in L^2(\R^3)$ , $U(t)$ is the free Schr\"odinger group and
$$\| r_{\infty} \|_{L^2} = O(t^{-1/4}) \quad \textrm{as} \;\; t \rightarrow
+\infty\ .$$
We stress the fact that in the present case and contrarily to the main
results in the field, the admitted nonlinearity is $L^2$-subcritical.  
%
\end{abstract}

\maketitle

\section{Introduction}
We continue here the analysis of a model of nonlinear Schr\"odinger equation with a concentrated nonlinearity in dimension three begun in \cite{ADO}.
We recall that such a model is defined by the equation
\begin{equation}    \label{eq1}
i\frac{du}{dt} = H_{\alpha(u)} u.
\end{equation}
where the nonlinear operator $H_{\alpha(u)}$ is a point interaction
(or ``delta potential'') in dimension three with a strength
$\alpha=\alpha(u)$ depending on the wavefunction $u$ in a prescribed
way. Here the nonlinearity is of power type and focusing. More
precisely, the domain of $H_{\alpha(u)}$ is given by
\be\label{domain}
D(H_{\alpha(u)}) = \left\{ u \in L^2(\R^3): \; u(x) = \phi(x) +q G_0(x) \,\ 
\textrm{with}\ \; \phi \in H^2_{loc}(\R^3), \, \Delta \phi \in
L^2(\R^3), \right.
\ee
$$\left. q \in \C,\ \ \lim_{x \rightarrow 0} (u(x) - q
G_0(x))=\phi(0) = - \nu |q|^{2 \sigma} q  \},\ \	
G_0(x) = \frac{1}{4\pi |x|}, \ \ \nu>0\ . \right\}
$$
\noindent and the action is given by \be\label{action}
H_{\alpha} u(x) = -\Delta \phi(x)\ ,\ \  x\in \R^3\setminus\{0\}\ . 
\ee

The nonlinearity, displayed in the boundary condition, means that the
value at zero of the ``regular part'' $\phi$ of the element domain is
related in a nonlinear way to the so called  ``charge'' $q$ of the
same element domain, which is the coefficent of the ``singular part''
$G_0$. According to \eqref{domain}, our choice for the function
$\alpha (u)$ is
$$
\alpha (u) \ = \ - \nu |q|^{2 \sigma}, \qquad \nu,\sigma > 0
$$ When $\sigma=0$ one obtains the well known contact interaction (of
strength $\alpha=-\nu$), which, due to the sign, is a so called
attractive $\delta$ interaction (see \cite{Albeverio}). When
$\sigma\neq 0$ the nonlinearity is in some sense acting at the single
point zero, coinciding with the location of the singularity of the
contact interaction. This is the origin of the denomination of
concentrated nonlinearity (see \cite{ADFT,ADFT2,NP}). The above model
can be derived from a standard nonlinear Schr\"odinger equationquation with a
inhomogeneous nonlinearity, i.e. with a nonlinearity space dependent,
and shrinking at a point in a suitable scaling limit. This derivation is
rigorously treated in \cite{CFNT1} for the one dimensional case and in
the forthcoming paper \cite{CFNT2} for the present three dimensional case. 
The problem \eqref{eq1} describes a Hamiltonian system for
which global well posedness holds in the nonlinearity range $\sigma\in
(0,1)$.
\noindent More precisely we endow the space $L^2(\R^3, \C)\simeq
L^2(\R^3, \R)\oplus L^2(\R^3, \R)$ (assuming the usual identification
$z\simeq (\Re z, \Im z)$) with the symplectic form 
\begin{equation}	\label{formasimplettica}
\Omega(u,v)=\Im \int_{\R^3} u{\overline v} \ dx = \int_{\R^3} (\Re v \Im u -\Im v \Re u) dx =\int_{\R^3} (u_2 v_1 -u_1 v_2) dx
\end{equation}
Of course $H^1(\R^3,\R)\oplus H^1(\R^3,\R)$ is a symplectic submanifold, and
we associate to \eqref{eq1} the Hamiltonian functional coinciding with
the total conserved energy 
associated to the evolution equation \eqref{eq1}, that is given by
\begin{equation}	\label{energy}
 E(u(t)) =\frac{1}{2}\|\nabla \phi\|_{L^2}^2 -\frac{\nu}{2\sigma +2}
|q|^{2\sigma +2}, \ \ u=\phi+qG_0 \in V.
\end{equation}
where $V$ is the domain of {\it finite energy states}
\begin{equation}    \label{def-V}
V = \{ u \in L^2(\R^3): \; u(x) = \phi(x) +q G_0(x), \, \textrm{with}
\; \phi \in L^2_{loc}(\R^3), \; \nabla \phi \in L^2(\R^3), \, q \in
\C \},
\end{equation}

\noindent which is a Hilbert space endowed with the norm
\begin{equation}	\label{norm-V}
\|u\|^2_V = \| \nabla \phi \|_{L^2}^2 + |q|^2.
\end{equation}

\noindent Note that for a generic element $u$ of the form domain the
charge $q$ and its regular part $\phi$ are independent of each other. 

\par\noindent
Correspondingly, the NLSE \eqref{eq1} can be rephrased in the hamiltonian form
\begin{equation}\label{eqham}
\frac{du}{dt}=J \ E^{'} (u(t))\ .
\end{equation} 
where $J = \left( \begin{array}{cc} 0 & 1 \\ -1 & 0 \end{array} \right)$ is the standard symplectic matrix.
In \cite{ADO} the case of power $\sigma\in (0,\frac{1}{\sqrt 2})$ was
studied,  showing existence of nonlinear bound states, orbital
stability and relaxation to (relative) equilibrium asymptotically in
time. For such values of the power nonlinearity $\sigma$ the
asymptotic analysis is simplified by the fact that linearization
around standing waves has no eigenvalues apart from the zero
eigenvalue, always existing due to the gauge $U(1)$ symmetry. For
$\sigma\in(\frac{1}{\sqrt 2},1)$ the linearization around a standing
wave admits pure imaginary eigenvalues $\pm i \xi$, which correspond
to existence of neutral oscillation in the linearized dynamics. If
these neutral oscillations persist as invariant tori in the phase
space of the complete nonlinear system, relaxation to an asymptotic
equilibrium standing wave is precluded. The analysis of so called
asymptotic stability of solitons started at the beginning of the nineties
with the studies of Soffer and Weinstein (\cite{SW1,SW2}) on the NLS
equation with an external potential in dimension three and of Buslaev
and Perelman on the translation invariant NLS in dimension one
(\cite{BP1,BP2}).  

\noindent
In our model the analysis goes as follows.
\noindent We consider the nonlinear evolution problem
\begin{equation}    \label{equation}
i\frac{du}{dt} = H_{\alpha} u, \qquad u(0)=u_0 \in D(H_{\alpha}),
\end{equation}

\par\noindent
In \cite{ADO} it is shown the existence of a solitary wave manifold for \eqref{equation}
\be\label{solitary}
\mathcal{M} = \left\{\ \Phi_{\omega}(x) = \left(
\frac{\sqrt{\omega}}{4\pi \nu} \right)^{\frac{1}{2\sigma}}
\frac{e^{-\sqrt{\omega} |x|}}{4\pi |x|} : \; \omega > 0 \,
\right\}\ ,
\ee
and it is shown that these solitary waves, which belong to
$D(H_{\alpha})$, are orbitally stable in the same range ($\sigma\in
(0,1)$) where global well posedness of equation \eqref{equation} is
guaranteed.\par\noindent 
One  aims at proving that an initial datum near this solitary manifold
relaxes for $t \to +\infty$ to the solitary manifold
itself.\par\noindent 
The solitary manifold turns out to be a symplectic two dimensional submanifold.\\
Now writing
$u=e^{i\omega t}(\Phi_\omega + R)$, we
obtain that $R$ satisfies the first order the linearized canonical system 
\be\label{linearizzato} 
\frac{dR}{dt} =  -J\left[
  \begin{array}{cc}
    H_{\alpha_1} + \omega & 0 \\
    0 & H_{\alpha_2} + \omega \\
  \end{array}
\right]R\  =\ \left[
  \begin{array}{cc}
    0 & L_2 \\
    -L_1 & 0 \\
  \end{array}
\right]R \ \equiv LR\ ,
\ee
and $L_j = H_{\alpha_j} +\omega$ for $j = 1,2$, 
$\alpha_1 = -(2\sigma +1) \frac{\sqrt{\omega}}{4 \pi}$ and $\alpha_2
= -\frac{\sqrt{\omega}}{4 \pi}$.

\noindent As recalled before (see \cite{ADO}, Section IV A, and
Appendix VI A of this paper) in the case $\sigma \in (1/\sqrt{2}, 
1)$ the discrete spectrum of $L$ consists of the eigenvalue $0$ with
algebraic multiplicity $2$ and of two purely imaginary eigenvalues $\pm
i \xi$ with
\begin{equation}\label{mu}
\xi = 2\sigma \sqrt{1 -\sigma^2} \omega,
\end{equation}
and corresponding eigenvectors $\Psi$ and $\Psi^*$ (see Appendix VI C of this paper).
\noindent As a consequence, the domain of the operator $L$ can be
decomposed in three symplectic subspaces, more precisely
$$D(L) = X^0 \oplus X^1 \oplus X^c,$$

\noindent where $X^0$, $X^1$, and $X^c$ are respectively the generalized kernel of $L$,
the eigenspace corresponding to the $L$ eigenfunctions $\Psi$ and
$\Psi^*$, and the subspace associated to the absolutely continuos spectrum of $L$.
In particular $X^0$ is generated by the tangent vectors to the symplectic solitary submanifold. 
The corresponding symplectic projection operators  from $L^2(\R^3)$ onto $X^0$,
$X^1$ and $X^c$ are denoted by $P^0$, $P^1$, $P^c$ respectively (see appendix C for explicit representation\ .)
\noindent A further step in the analysis consists in decomposing the solution of equation \eqref{equation} in the sum of a soliton-like part $e^{i\Theta(t)} \Phi_{\omega(t)}(x)$ and a fluctuating part $\chi(t,x)$, introducing the Ansatz 
\begin{equation}    \label{mod-ansatz}
u(t,x) = e^{i\Theta(t)} \left( \Phi_{\omega(t)}(x) +\chi(t,x)
\right),
\end{equation}
 with
\be\label{phase}\Theta(t) = \int_0^t \omega(s) ds +\gamma(t),\ee
and
\begin{equation} \label{zete}
\chi(t,x) = z(t) \Psi(t,x) +\overline{z(t)} \Psi^*(t,x) +f(t,x) \equiv
\psi(t,x) +f(t,x), \ \ \psi \in X^1,\ \ f \in X^c
\end{equation}
\noindent with $\omega(t)$, $\gamma(t)$, $z(t)$ and $f(t,x)$ up to now
unprejudiced. Notice that the parameters are now time dependent, as
well as the functions $\Psi$ and $\overline{\Psi}$: this is due to the
fact that $\Psi$ depends on $t$ through $L$ and then through $\omega$.
  

\noindent
The goal now is to show that the fluctuating part is decaying
($\Psi$-$\Psi^*$ component) or dispersing ($f$-component), to provide
convergence of parameters $\omega(t), \gamma(t)$ to possibly unknown
asymptotic values; all of this finally gives relaxation to the solitary manifold. 
To this end one has to fix the above underdetermined representation,
and then obtain equations for $\omega(t)$, $\gamma(t)$ and
$\chi(t,x)$.  
This is achieved by requiring the fluctuation $\chi$ to be
symplectically orthogonal to the solitary manifold $\mathcal{M}$ (or,
equivalently, orthogonal to the generalized kernel $N_g(L)$) for every
time $t\geq 0$. This procedure yields the so called ``modulation equations'' for
$\omega(t)$, $\gamma(t)$ and $\chi(t,x)$. Moreover, exploiting further
orthogonality relations between $\Psi,\Psi^*,\Phi_{\omega}$, one also
gets equations for the coefficient $z$ and the dispersive term $f$ in
$\chi$ (See Theorem \ref{lemma-mod} in Section \ref{modeq} for more
details). 
At this point two problems occur. The first is that the modulation
equation for the fluctuating component $\chi$, due to the introduction
of the time dependent Ansatz,  contains a non autonomous linear
part. So, the use of dispersive estimates to show the wanted decay of
this component requires to preliminarily ``freeze'' the dynamics at a
certain fixed time $T$; the subsequent step is to show the uniform character of the
estimates in $T$. The frozen equations are written in
Section \ref{modeq}. In the same section the leading terms in the
modulation equations are identified, and the estimates on the
remainder terms are displayed. This is however not yet sufficient to
get rid of the complex oscillation described by $z$ and $z^*$. So the
modulation equations are rewritten in a canonical way by means of a
Poincar\'e normal form pushed to the third order which preserves
estimates on the remainders. In particular, the transformed oscillating
component $z_1=z_1(z,\bar z)$ satisfies the equation 
\begin{equation}\label{zuno}
\dot{z_1} = i\xi z_1 +i K |z_1|^2 z_1 +\widehat{Z}_R\ .
\end{equation}
The asymptotic behavior of the solution of the $z_1$ component depends
on the coefficient $iK$ and more precisely on the sign of its real
part.\\ 
This is the point where nonlinear Fermi golden rule (FGR) enters the
game. It turns out in relevant examples that if a resonance condition
between a higher harmonic (a multiple) of $i\xi$ and the continuous
spectrum of linearization is satisfied (the FGR), then $\Re (i K)$ is
strictly negative; this gives decay of the oscillating modes of the
linearization (see the seminal paper cited above and moreover
\cite{TY}, \cite{TY2} and \cite{BS}). The decay is ultimately due to
coupling of oscillating modes with the continuous spectrum assured by
FGR, and consequent drift of energy from the discrete component to the
continuous one; so the mechanism is dissipation by dispersion. Furthermore,
exploiting the hamiltonian structure of the system it can be shown that
the above situation is in some sense generic (\cite{CM, Cu1},\cite{B}
and reference therein).  

In the present model things go as follows. In the first place the second harmonic $2i\xi$ of the discrete eigenvalue $i\xi$ lies in the interior of the continuous spectrum if the nonlinearity power $\sigma\in (\frac{1}{\sqrt{2}}, \frac{\sqrt{3} +1}{2 \sqrt{2}})$.
Let us then denote by
$N_2(q, q)$ the quadratic terms coming from the Taylor expansion of
the nonlinearity (see formula \eqref{quadratic-remainder}), and by $\Psi(\omega_0) = \left(
\begin{array}{cc}
\Psi_1(\omega_0)\\
\Psi_2(\omega_0)
\end{array} \right)$ the eigenfunction of the linearized operator associated to
$i\xi_0$ (given explicitely in Appendix VI C, Proposition VI 4). The nonlinear Fermi golden rule (FGR) adapted to our case is the following non-degeneracy condition:
\begin{equation}    \label{FGR}
JN_2(q_{\Psi(\omega_0)}, q_{\Psi(\omega_0)}) \overline{q_{\Psi_+(2i
\xi_0)}} \neq 0,
\end{equation}

\noindent where $\Psi_+(2i \xi_0)$ is the generalized eigenfunction
associated to $+2i \xi_0$ (for the explicit representation see Appendix D).\\
This nondegeneracy condition should imply $\Re (i K)<0$ in
\eqref{zuno}.\\ As a matter of fact, thanks to the explicit character
of our model, we are able to 
directly verify, without use of the nondegeneracy condition, that $\Re
(i K)<0$ (that gives dissipation in \eqref{zuno}) holds for any $\sigma$ 
in the range $\left( \frac{1}{\sqrt{2}}, \sigma^* \right)$, for a
certain $\sigma^* \in \left( \frac{1}{\sqrt{2}}, \frac{\sqrt{3}
+1}{2\sqrt{2}} \right]$ (see Section \ref{canonical-z}). Moreover, we
  have
numerical evidence that this is true on the whole interval
$\left( \frac{1}{\sqrt{2}}, \frac{\sqrt{3} +1}{2\sqrt{2}} \right)$.
With these premises, we eventually prove the following result.

\noindent {\bf Theorem} {\textbf{(Asymptotic stability in the case
of purely imaginary eigenvalues)}} Assume that $u \in C(\R^+, V)$
is a solution to equation \eqref{equation} with a power nonlinearity
(see \eqref{domain}) given by $\sigma \in (\frac{1}{\sqrt{2}}, 
\sigma^*)$, for a certain $\sigma^* \in
(\frac{1}{\sqrt{2}}, \frac{\sqrt{3} +1}{2 \sqrt{2}}]$. Moreover,
suppose that the initial datum is close to a standing wave of
\eqref{equation}, in the sense that  
$$u(0) = u_0 = e^{i\omega_0 +\gamma_0} \Phi_{\omega_0} +e^{i\omega_0
+\gamma_0} [(z_0 +\overline{z_0}) \Psi_1 +i (z_0 -\overline{z_0})
\Psi_2] +f_0 \in V \cap L^1_w(\R^3),$$

\noindent with $\omega_0 > 0$, $\gamma_0$, $z_0 \in \R$, and $f_0
\in L^2(\R^3) \cap L^1_w(\R^3)$ and
$$|z_0| \leq \epsilon^{1/2} \qquad \textrm{and} \qquad \|f_0\|_{L^1_w} \leq
c \epsilon^{3/2},$$

\noindent where $c$, $\epsilon > 0$.

\noindent Then, provided $\epsilon$ is sufficiently small, the
solution $u(t)$ can be asymptotically decomposed as
$$u(t) = e^{i\omega_{\infty} t +i b_1 \log (1 +\epsilon k_{\infty} t)}
\Phi_{\omega_{\infty}} +U_t*\psi_{\infty} +r_{\infty}, \quad
\textrm{as} \;\; t \rightarrow +\infty,$$

\noindent where $\omega_{\infty}$, $\epsilon k_{\infty} > 0$, $b_1
\in \R$, and $\psi_{\infty}$, $r_{\infty} \in L^2(\R^3)$ such that
$$\| r_{\infty} \|_{L^2} = O(t^{-1/4}) \quad \textrm{as} \;\; t \rightarrow
+\infty,$$

\noindent in $L^2(\R^3)$.

\vskip5pt
Some remarks are in order.\\
\noindent We stress again that the range of the admitted nonlinearities
$\sigma$ implies that $\pm 2i \xi$ is in the essential spectrum of
the linearized operator. Notice that $\frac{\sqrt{3}
  +1}{2\sqrt{2}}\approx 0.96$, and that when $\sigma\to 1$ the
discrete eigenvalues of linearization, $\pm i\xi$, collide at zero.
So, in order  to cover a larger range of values of $\sigma$ one has to consider harmonics higher than the second and consider normal form pushed to order higher than the third. As shown in \cite{ADO} the standing waves in the case $\sigma=1$ are unstable by blow-up.\\
We recall that $\sigma = \frac{1}{{\sqrt 2}}$ is a threshold resonance for the linearization. Its presence does not allow enough decay of the linearized dynamics to obtain relaxation to solitary manifold.\\
The asymptotic stability result is achieved following the outline of
\cite{BS} and \cite{KKS} and using the machinery already developed in \cite{ADO}. In particular, in \cite{KKS} the same
problem for the analogous one-dimensional model is studied.

\noindent Nevertheless, the three-dimensional case presents some
differences. The first one is that the concentrated nonlinearity imposes to develop the
analysis at the form level. This means that the estimates on the
evolution of the initial data are more delicate.

\noindent The second main difference is the faster decay of the
propagator of the free Laplacian. This allows to develop the the
analysis using just the structural weight $w = 1 +\frac{1}{|x|}$
which arises from the dispersive estimate instead of introducing new
weighted spaces as done 
in the one-dimensional case treated in \cite{BKKS, KKS}.

\noindent Finally, the eigenfunctions associated to the purely
imaginary eigenvalues do not have oscillating terms occurring in the
one-dimensional case but they exponentially decrease as $|x|
\rightarrow +\infty$. This fact will be useful in order to get
the decay in time of the radiation term.

\noindent A last comment of general nature is in order. As in the
one dimensional case studied by Buslaev, Komech, Kopylova, and
Stuart in \cite{BKKS} and Komech, Kopylova, and Stuart in
\cite{KKS},
and the three dimensional model analyzed in \cite{ADO}, the analysis
of a specific model allows to obtain 
asymptotic stability of standing waves without a priori assumptions.
In particular the nonlinearity is fixed, of power type and
subcritical, in the sense that it falls in the range of global well
posedness of the equation (see \cite{KM} for a different example where
asymptotic stability is proven in subcritical regime); no spectral
assumptions are needed; no smallness of initial data is required, in
the sense 
that we give results for every standing wave of the model and
initial data near the family of standing waves.\\ As a final remark, while
Komech, Kopylova, and Stuart in \cite{KKS} find a link between the Fermi Golden
Rule and the decay of normal modes of linearization, here such decay is directly
verified. This fact seems to indicate that some of these assumptions
or hypotheses are in fact unnecessary when enough information about
the model is known.
\par\noindent For the sake of completeness, in the course of the paper we will
repeat proofs requiring some modifications because of the facts
mentioned above; on the contrary, where the arguments hold
unchanged, just a reference will be given. Moreover in the appendices we give information about linearization operator recalling useful material from \cite{ADO} and giving detailed properties of discrete and continuous eigenfunctions.
\vskip10pt
\par\noindent
{\bf Acknowledgments} The authors are grateful to Scipio Cuccagna, Gianfausto
Dell'Antonio, Alexander Komech and Galina Perelman for several discussions and correspondence. 
R.A. and D.N. are partially supported by 
 the FIRB 2012 grant {\em Dispersive Dynamics: Fourier Analysis and
   Variational Methods}, reference number RBFR12MXPO.
R.A. is member of the INdAM Gruppo Nazionale per l'Analisi
Matematica, la Probabilit\`a e le loro Applicazioni; D. N.  is member
of the INdAM Gruppo Nazionale per la Fisica Matematica.






\section{Modulation equations}\label{modeq}
We recall that
the operators we are dealing with have different domains
while the forms associated share the common domain $V$, introduced
in \eqref{def-V}.
It proves convenient to describe $V$ in an alternative way: fixed an
arbitrary $\lambda > 0$ and denoted $G_\lambda (x) : = e^{-\lambda
  |x|}/4\pi$, one finds
$$V = \left\{ u = \phi_{\lambda}+ q G_{\lambda}, \,
\textrm{with} \, \phi_{\lambda} \in H^1(\R^3), \, q \in \C
\right\}.$$
The arbitrary positive parameter $\lambda$ is customary in the description of a point interaction and moreover allows the regular part of the element's domain to be in  the Sobolev space $H^1$ instead of the corresponding homogeneous counterpart, which is often useful. Everything is in fact independent on the choice of $\lambda\ .$ \\
Correspondingly, we will perform computations at the
form level. In order to do that let
us recall that the variational formulation of equation \eqref{eq1}
is
\begin{equation}    \label{eq-variazionale}
\left( i\frac{du}{dt}(t), v \right)= Q_{\alpha}(u(t),v) \quad
\forall v \in V.
\end{equation}
where, given $u,v \in V$ with $u =  \phi_{u, \lambda} + q_u
G_\lambda$,  $v =  \phi_{v, \lambda} + q_v
G_\lambda$,
$$Q_\alpha (u,v) = 
( \nabla \phi_{u, \lambda},  \nabla \phi_{v, \lambda})_{L^2} + \bar q_u q_v
\left( \alpha + \frac {\sqrt \lambda}{8 \pi} \right) - \lambda \bar q_u (
G_\lambda, \phi_{v, \lambda})_{L^2} - \lambda q_v (
\phi_{u, \lambda}, G_\lambda)_{L^2}
$$

\noindent Note that the equation \eqref{eq-variazionale}
makes sense because $V$ is
independent of the positive parameter $\lambda$ and it is a Hilbert
space with the norm
$$\|u\|_V^2 = \| \nabla \phi_{\lambda} \|_{L^2}^2 +|q|^2, \quad \forall u \in
V.$$

\noindent In order to inspect the asymptotic stability of equation
\eqref{eq1} it is useful to represent the solution $u$
by the aid of
the Ansatz \eqref{mod-ansatz} together with defintions
\eqref{phase} and \eqref{zete}. From now on we always refer to such formulas.




\noindent Hence, we want to construct a solution of
equation \eqref{eq1} close at each time to a solitary wave. Let us
notice that the solitary wave does not need to be the same at every
time, which means that the parameters $\omega(t)$ and $\Theta(t)$
are free to vary in time.

\noindent Exactly as in the case $\sigma \in \left( 0,
\frac{1}{\sqrt{2}} \right)$ (see \cite{ADO}, Section V) the function
$\chi$ solves
\begin{equation}    \label{eqchi}
\left( i\frac{d\chi}{dt}(t), v \right)_{L^2} = Q_{\alpha,
Lin}(\chi(t),v) +\dot{\gamma}(t) (\Phi_{\omega(t)} +\chi(t),
v)_{L^2} +
\end{equation}

$$+\dot{\omega}(t) \left(-i\frac{d\Phi_{\omega(t)}}{d\omega}
,v\right)_{L^2} +N(q_{\chi}(t), q_v),$$

\noindent for all $v \in V$.
I

Here $Q_{\alpha,Lin}$ is the quadratic
form of  the linearization operator acting as
$$Q_{\alpha,Lin}(\chi,v) = (\nabla \phi_{\chi}, \nabla \phi_v)_{L^2}
-\frac{\sqrt{\omega}}{4\pi} \Re (q_{\chi} \overline{q_v})
-\sigma \frac{\sqrt{\omega}}{2\pi} \Re q_{\chi} \Re q_v +\omega (\chi,v)_{L^2}\ , $$

\noindent and the nonlinear remainder $N(q_{\chi},q_v)$ is given by
\begin{equation}\label{quadratic-remainder}
N(q_{\chi}, q_v) = -\nu |q_{\chi}
+q_{\omega}|^{2\sigma} \Re ((q_{\chi} +q_{\omega}) \overline{q_v})
+\nu (2\sigma +1)|q_{\omega}|^{2\sigma} \Re q_{\chi} \Re q_v +\nu
|q_{\omega}|^{2\sigma} \Im q_{\chi} \Im q_v +\nu
|q_{\omega}|^{2\sigma} \Re(q_{\omega} \overline{q_v}),
\end{equation}
where (see section II B in \cite{ADO}), $q_\omega =
\left( \frac {\sqrt \omega} {4 \pi \nu} \right)^{\frac 1 {2 \sigma}}.$

\noindent Since $\omega(t)$, $\gamma(t)$, and $\chi(x,t)$ are
unknown and the propagator grows in time along the directions of the
generalized kernel of the operator $L$, the idea is to get a
determined system requiring the function $\chi(t)$ to be orthogonal
to the generalized kernel of $L$ at any time $t \geq 0$. Hence, one
obtains that $\omega$, $\gamma$, $z$, and $f$ must solve the
following system of equations.
\begin{theo}{\emph{\textbf{(Modulation equations)}}}  \label{lemma-mod}
If $\chi(t)$ is a solution of equation \eqref{eqchi} such that $P_0
\chi(t) = 0$ for all $t \geq 0$ and $\omega(t)$ and $\gamma(t)$ are
continuously differentiable in time, then $\omega$ and $\gamma$ are
solutions of
\begin{equation}    \label{mod-eq1}
\dot{\omega} = \frac{\Re\ (JN(q_{\chi})
\overline{q_{P_0^*(\Phi_\omega+\chi)}})}{\left( \varphi_{\omega} -
\frac{d P_0}{d\omega}\chi , \Phi_\omega + \chi \right)_{L^2}},
\end{equation}

\begin{equation}    \label{mod-eq2}
\dot{\gamma} = \frac{\Re\ (JN(q_{\chi})
\overline{q_{J({\varphi_\omega}-\frac{d P_0}{d\omega}\chi)}})}
{\left( \varphi_{\omega} - \frac{d P_0}{d\omega}\chi , \Phi_\omega +
\chi \right)_{L^2}}.
\end{equation}

\noindent Furthermore, $z$ and $f$ satisfy

\begin{equation}    \label{mod-eq3}
(\Psi, J\Psi)_{L^2} (\dot{z} -i\xi z) = \Re (JN(q_{\chi})
\overline{q_{J\Psi}}) +\dot{\omega} \left[ \left( f,
J\frac{d\Psi}{d\omega} \right)_{L^2} -\left( \frac{d\psi}{d\omega},
J\Psi \right)_{L^2} \right] +\dot{\gamma} (\chi, \Psi)_{L^2},
\end{equation}

\begin{equation}    \label{mod-eq4}
\left( \frac{df}{dt}, v \right)_{L^2} = Q_L(f,v) +\left(
-\dot{\omega} \left( z P^c \frac{d\Psi}{d\omega} +\overline{z} P^c
\frac{d\Psi^*}{d\omega} \right) +\dot{\gamma} P^c J \chi, v
\right)_{L^2} +
\end{equation}
$$+(8\pi \sqrt{\lambda} P^c JN(q_{\chi}) G_{\lambda}, q_v
G_{\lambda})_{L^2},$$

\noindent for all $v \in V$.
\end{theo}

\begin{proof}
Equations \eqref{mod-eq1} and \eqref{mod-eq2} can be proved with the
same argument exploited in the case $\sigma \in \left( 0,
\frac{1}{\sqrt{2}} \right)$ (see \cite{ADO}, Theorems V.3 and V.6).

\noindent Equation \eqref{mod-eq3} can be obtained taking $v =
J\Psi$ as test function and noting that
\begin{itemize}
  \item $\frac{d\chi}{dt} = \dot{z} \Psi +\dot{\overline{z}} \Psi^*
+\dot{\omega} \left( z \frac{d\Psi}{d\omega} +\overline{z}
\frac{d\Psi^*}{d\omega} \right) +\frac{df}{dt}$,
  \item $(\Psi^*, J\Psi)_{L^2} = 0$,
  \item $\left( \frac{d\Phi_{\omega}}{d\omega}, J\Psi \right)_{L^2} = 0$,
  \item $\left( \frac{df}{dt}, J\Psi \right)_{L^2} = -\dot{\omega} \left( f,
J\frac{d\Psi}{dt} \right)_{L^2}$, and
  \item $\dot{\omega} \left( \frac{d\Psi^*}{d\omega}, J\Psi \right)_{L^2} =
-\left( \Psi^*, J\frac{d\Psi}{dt} \right)_{L^2}$.
\end{itemize}

\noindent Finally, equation \eqref{mod-eq4} follows taking the
projection onto the continuous spectrum $P^c$ of both sides of
equation \eqref{eqchi} and recalling that $f \in X^c$.
\end{proof}

\subsection{Frozen spectral decomposition}
The goal of this subsection is to get an autonomous linearized
equation for the component $f$. 

According to \eqref{linearizzato}, having fixed $T>0$ we define
$$L_T = L(\omega_T) : =  -J\left[
  \begin{array}{cc}
    H_{\alpha_1(T)} + \omega_T & 0 \\
    0 & H_{\alpha_2(T)} + \omega_T \\
  \end{array}
\right],$$
where $\alpha_1 (T) = - (2 \sigma + 1) \frac{\sqrt{\omega_T}}{4 \pi}$, $\alpha_2 = -  \frac{\sqrt{\omega_T}}{4 \pi}$. 
and $\omega_T = \omega(T)$.
\noindent Then for any $t \in [0, T]$ one
can decompose $f(t) \in X^c = X^c(t)$ as
$$f = g +h \qquad \textrm{with} \quad g \in X^d_T = X^0_T \oplus X^1_T, \; h
\in X^c_T,$$

\noindent where the subscript $T$ means that the time is fixed at $t
= T$.
We recall that $P^0$ and $P^1$ are the symplectic projections on the kernel and nonzero eigenvalues of the linearization (see Appendix C). Correspondingly we denote by
$P^0_T$ and $P^1_T$ the projections on the analogous subspaces frozen at time $T$. Finally denote $P^d_T = P^0_T +P^1_T$. 


\noindent Then for the quadratic form
$$Q_L(u,v) -Q_{L_T}(u,v) =
\frac{\sqrt{\omega} -\sqrt{\omega_T}}{4\pi} \Re (\mathbb{T} q_u
\overline{q_v}) -(\omega_T -\omega) (Ju, v)_{L^2},$$

\noindent for all $u$, $v \in V$, where
$$\mathbb{T} =\left[
  \begin{array}{cc}
    0 & -1 \\
    2\sigma +1 & 0 \\
  \end{array}
\right].$$

\noindent Hence, observing that $P^c \Psi = 0$, the equation
\eqref{mod-eq4} for $f$ is equivalent to
\begin{equation}\left( \frac{df}{dt}, v \right)_{L^2} = Q_{L_T}(f,v) +\left( (\omega
-\omega_T) Jf +\dot{\omega} \frac{dP^c}{d\omega} \psi +\dot{\gamma}
P^c J \chi, v \right)_{L^2} +
\label{label}
\end{equation}
$$+\left( 8\pi \sqrt{\lambda} \left( \frac{\sqrt{\omega} -\sqrt{\omega_T}}{4\pi}
\mathbb{T} q_f +P^c JN(q_{\chi}) \right) G_{\lambda}, q_v
G_{\lambda} \right)_{L^2},$$

\noindent for all $v \in V$.

\noindent Since our dispersive estimate holds only on the continuous
spectral subspace, we need to prove that it is enough to estimate
the symplectic projection of $\chi(t)$ onto that subspace. This is
stated in the following lemma where we denote
by $\mathcal{R}(a)$ a  bounded continuous real
valued functions vanishing as $a \to 0$, and
$$\mathcal{R}_1(\omega) = \mathcal{R}(\|\omega
-\omega_0\|_{C^0([0,T])}).$$

\noindent
Along the paper, with a
slight abuse, we shall use
the symbol
 $\mathcal{R}(a,b)$ to denote bounded continuous real
valued functions vanishing as $a$, $b \rightarrow 0$.

\begin{lemma}   \label{lemma-cont}
If $|\omega -\omega_T|$ is small enough, then the function $g$ can
be estimated in terms of $h$ as follows:
$$\|g\|_{L^{\infty}_{w^{-1}}} \leq \mathcal{R}_1(\omega)
|\omega -\omega_T| \|h\|_{L^{\infty}_{w^{-1}}}.$$
\end{lemma}

\noindent The last lemma can be proved following the proof of Lemma
3.2 in \cite{KKS}.




\noindent As a consequence, one can apply the operator $P^c_T$ to
both sides of the equation for $f$ and obtain
\begin{equation}    \label{eq-h}
\left( \frac{dh}{dt}, v \right)_{L^2} = Q_{L_T}(h,v) +\left( P^c_T
\left[ (\omega -\omega_T) Jf +\dot{\omega} \frac{dP^c}{d\omega} \psi
+\dot{\gamma} P^c J \chi \right], v \right)_{L^2} +
\end{equation}

$$+\left( 8\pi \sqrt{\lambda} P^c_T \left( \frac{\sqrt{\omega}
-\sqrt{\omega_T}}{4\pi}\mathbb{T} q_f +P^c JN(q_{\chi}) \right)
G_{\lambda}, q_v G_{\lambda} \right)_{L^2},$$

\noindent for any $v \in V$.

\subsection{Asymptotic expansion of dynamics}
In order to prove the asymptotic stability of the ground state we
need to show that for large times $z$ and $h$ are small. For this
purpose, in this section we expand the inhomogeneous
terms in the modulation equations, then rewrite the equations of
$\omega$, $\gamma$, $z$ and $h$ as in \eqref{sv-omega},
\eqref{sv-gamma}, \eqref{sv-z} and \eqref{sv-h}, and for each equation
we estimate the error terms.

\noindent In what follows we denote
$$(q, p) = q_1 p_1 +q_2 p_2, \qquad \forall p,q \in \C^2.$$

\noindent With an abuse of notation we denote by
$q_{\omega} = \left( \begin{array}{cc}
\left( \frac{\sqrt{\omega}}{4\pi \nu} \right)^{1/(2\sigma)}\\
0
\end{array} \right)$ the charge of the function $\left( \begin{array}{cc}
\Phi_{\omega}\\
0
\end{array} \right)$.

\noindent As a preliminary step, we expand the nonlinear part of the
equation \eqref{eqchi} $N(q_{\chi})$ as
\begin{equation}    \label{N}
N(q_{\chi}) = N_2(q_{\chi}) +N_3(q_{\chi}) +N_R(q_{\chi}),
\end{equation}

\noindent where $N_2$ and $N_3$ are the quadratic and cubic terms in
$q_{\chi}$ respectively, while $N_R$ is the remainder. Exploiting
the Taylor expansion of the function $F(t) = t^{\sigma}$ around
$|q_{\omega}|^2$, one gets
$$\Re (N_2(q_{\chi}) \overline{q_v}) = \Re ((\sigma |q_{\omega}|^{2(\sigma -1)}
|q_{\chi}|^2 q_{\omega} +2\sigma |q_{\omega}|^{2(\sigma -1)}
(q_{\omega}, q_{\chi}) q_{\chi} +2(\sigma -1) \sigma
|q_{\omega}|^{2(\sigma -2)} (q_{\omega}, q_{\chi})^2 q_{\omega})
\overline{q_v}),$$

\noindent and
$$\Re (N_3(q_{\chi}) \overline{q_v}) = \Re (( \sigma |q_{\omega}|^{2(\sigma -1)}
|q_{\chi}|^2 q_{\chi} +2(\sigma -1) \sigma |q_{\omega}|^{2(\sigma
-2)} (q_{\omega}, q_{\chi})^2 q_{\chi} +$$
$$+2(\sigma -1) \sigma |q_{\omega}|^{2(\sigma -2)} (q_{\omega},
q_{\chi}) |q_{\chi}|^2 q_{\omega} +\frac{4}{3} (\sigma -2) (\sigma
-1) \sigma |q_{\omega}|^{2(\sigma -3)} (q_{\omega}, q_{\chi})^3
q_{\omega}) \overline{q_v}),$$

\noindent for any $q_v \in \C$. For later convenience, let us define
the following symmetric forms
$$N_2(q_1, q_2) = \sigma |q_{\omega}|^{2(\sigma -1)} (q_1, q_2) q_{\omega}
+\sigma |q_{\omega}|^{2(\sigma -1)} [(q_{\omega}, q_1) q_2
+(q_{\omega}, q_2) q_1] +$$
$$+2(\sigma -1) \sigma |q_{\omega}|^{2(\sigma
-2)} (q_{\omega}, q_1) (q_{\omega}, q_2) q_{\omega},$$

\noindent and
$$N_3(q_1, q_2, q_3) = \frac{1}{6} \sigma |q_{\omega}|^{2(\sigma -1)}
\sum_{i,j,k=1}^3 (q_i, q_j) q_k +\frac{1}{3} (\sigma -1) \sigma
|q_{\omega}|^{2(\sigma -2)} \sum_{i,j,k=1}^3 (q_{\omega}, q_i)
(q_{\omega}, q_j) q_k +$$
$$+\frac{1}{3} (\sigma -1) \sigma |q_{\omega}|^{2(\sigma -2)}
\sum_{i,j,k=1}^3 (q_{\omega}, q_i) (q_j, q_k) q_{\omega}
+\frac{4}{3} (\sigma -2) (\sigma -1) \sigma |q_{\omega}|^{2(\sigma
-3)} (q_{\omega}, q_1) (q_{\omega}, q_2) (q_{\omega}, q_3)
q_{\omega}.$$

\noindent In order to prove the asymptotic stability result, we
shall prove in Section 2.4,  the following asymptotics
\begin{equation}    \label{asintoti}        
\| f(t) \|_{L^{\infty}_{w^{-1}}} \sim t^{-1}, \qquad z(t) \sim
t^{-\frac{1}{2}}, \qquad \| \psi(t) \|_{V} \sim t^{-\frac{1}{2}},
\end{equation}

\noindent as $t \rightarrow +\infty$.
\begin{rmk}
As in \cite{KKS}, the first step in proving these expected
asymptotics is to separate leading terms and remainders in the right
hand sides of the modulation equations \eqref{mod-eq1} -
\eqref{mod-eq3}, \eqref{eq-h}. Basically, in the next subsections,
we will expand the expression for $\dot{\omega}$, $\dot{\gamma}$,
and $\dot{z}$ up to and including the terms of order $t^{-3/2}$, and
for $\dot{h}$ up to and including $t^{-1}$.
\end{rmk}

\begin{rmk}
Note that since the nonlinearity depends only on the charges the
same holds for its Taylor expansion.
\end{rmk}

\subsubsection{Equation for $\omega$}
Substituting the expansion for the nonlinear part $N$ given in
\eqref{N} in equation \eqref{mod-eq1} and considering the
asymptotics \eqref{asintoti} one gets
$$\dot{\omega} = \frac{1}{\Delta} \Re ((JN_2(q_\psi) +2JN_2(q_\psi, q_f)
+JN_3(q_\psi)) q_{\omega}) +\frac{1}{\Delta^2} \left( \psi,
\frac{d\Phi_{\omega}}{d\omega} \right)_{L^2} \Re (JN_2(q_\psi)
q_{\omega}) +\Omega_R,$$

\noindent where $\Delta = \frac{1}{2} \frac{d}{d \omega} \|
\Phi_{\omega} \|_{L^2}^2$ and the remainder $\Omega_R$ is estimated
by
$$|\Omega_R| \leq \mathcal{R}(\omega, |z| +\|f\|_{L^{\infty}_{w^{-1}}})
(|z|^2 +\|f\|_{L^{\infty}_{w^{-1}}})^2.$$      

\noindent Recalling that $\psi = z\Psi +\overline{z} \Psi^*$, one
can rewrite the previous equation for $\dot{\omega}$ as
\begin{equation}    \label{sv-omega}
\dot{\omega} = \Omega_{20}z^2 +\Omega_{11}z\overline{z}
+\Omega_{02}\overline{z}^2 +\Omega_{30}z^3 +\Omega_{21}z^2
\overline{z} +\Omega_{12}z\overline{z}^2 +\Omega_{03}\overline{z}^3
+z(q_f, \Omega'_{10}) +\overline{z}(q_f, \Omega'_{01}) +\Omega_R,
\end{equation}
where the $\Omega_{ij}$'s and the $\Omega_{ij}'$'s are suitable
coefficients, while $\Omega_R$ is a remainder term.

\begin{rmk}
Since the second component of the vector $q_{\omega}$ equals $0$,
one has
$$\Omega_{11} = 2\frac{q_{\omega}}{\Delta} Re (JN_2(q_{\Psi})
\overline{q_{\Psi^*}}) = 0.$$

\noindent This fact will turn out to be useful in writing the
canonical form of the modulation equations.
\end{rmk}

\subsubsection{Equation for $\gamma$}
As in the previous subsection the equation for $\dot{\gamma}$
\eqref{mod-eq2} can expanded as
$$\dot{\gamma} = \frac{1}{\Delta}
\Re ((JN_2(q_\psi) +2JN_2(q_\psi, q_f) +JN_3(q_\psi))
\overline{q_{J\frac{d\Phi_{\omega}}{d\omega}}}) +\frac{1}{\Delta^2}
\left( \psi, J\frac{d^2\Phi_{\omega}}{d^2\omega} \right)_{L^2} \Re
(JN_2(q_\psi) q_{\omega}) +\Gamma_R,$$

\noindent where the remainder $\Gamma_R$ is estimated by
$$|\Gamma_R| \leq \mathcal{R}(\omega, |z| +\|f\|_{L^{\infty}_{w^{-1}}})
(|z|^2 +\|f\|_{L^{\infty}_{w^{-1}}})^2.$$

\noindent As before, the equation for $\dot{\gamma}$ shall be
written in the form
\begin{equation}    \label{sv-gamma}
\dot{\gamma} = \Gamma_{20}z^2 +\Gamma_{11}z\overline{z}
\Gamma_{02}\overline{z}^2 +\Gamma_{30}z^3 +\Gamma_{21}z^2
\overline{z} +\Gamma_{12}z\overline{z}^2 +\Gamma_{03}\overline{z}^3
+z(q_f, \Gamma'_{10}) +\overline{z}(q_f, \Gamma'_{01}) +\Gamma_R.
\end{equation}

\begin{rmk}
In this case $\Gamma_{11}$ does not vanish as in equation
\eqref{sv-omega}.
\end{rmk}

\subsubsection{Equation for $z$}
Exploiting the results of the previous subsections, equation
\eqref{mod-eq3} can be expanded as
$$\begin{array}{ll}
\dot{z} -i\xi z & = \frac{2}{\kappa} \Re (JN_2(q_\psi)
\overline{q_f}) +\frac{1}{\kappa} \Re ((JN_2(q_\psi)
+JN_3(q_\psi)) \overline{q_{J\Psi}}) +\\
& -\frac{1}{\Delta \kappa} \left( \frac{d\psi}{d\omega}, J\Psi
\right)_{L^2} \Re (JN_2(q_\psi) q_{\omega}) +\frac{1}{\Delta \kappa}
(\psi, \Psi)_{L^2} \Re (JN_2(q_\psi)
\overline{q_{J\frac{d\Phi_{\omega}}{d\omega}}}) +Z_R,
\end{array}$$

\noindent where $\kappa = -(\Psi, J\Psi)_{L^2}$ and
$$|Z_R| \leq \mathcal{R}(\omega, |z| +\|f\|_{L^{\infty}_{w^{-1}}})
(|z|^2 +\|f\|_{L^{\infty}_{w^{-1}}})^2.$$

\noindent Mimicking the notation employed in \eqref{sv-omega} and \eqref{sv-gamma}, the previous equation
can be written in the form
\begin{equation}    \label{sv-z}
\dot{z} = i\xi z +Z_{20}z^2 +Z_{11}z\overline{z}
+Z_{02}\overline{z}^2 +Z_{30}z^3 +Z_{21}z^2 \overline{z}
+Z_{12}z\overline{z}^2 +Z_{03}\overline{z}^3 +z \Re (q_f \overline{Z'_{10}})
+\overline{z} \Re (q_f \overline{Z'_{01}}) +Z_R,
\end{equation}

\noindent and it turns out that
\begin{equation}    \label{eqZij}
\begin{array}{ll}
Z_{11} =&\frac{2}{\kappa} \Re(JN_2(q_{\Psi}, q_{\Psi^*}) \overline{q_{\Psi}}), \quad Z_{20} = \frac{1}{\kappa} \Re(JN_2(q_{\Psi}) \overline{q_{\Psi}}), \quad Z_{02} = \frac{1}{\kappa} \Re(JN_2(q_{\Psi^*}) \overline{q_{\Psi}}),\\

Z_{21} =&\frac{3}{\kappa} \Re(JN_3(q_{\Psi^*}, q_{\Psi}, q_{\Psi}) \overline{q_{\Psi}})+ \frac{1}{\Delta \kappa} \left[ \left( \frac{d \Psi^*}{d \omega}, j\Psi \right)_{L^2} \Re(JN_2(q_{\Psi}) \overline{q_{J\Phi_{\omega}}}) +\right.\\

& \left. -(\Psi^*, \Psi)_{L^2} \Re(JN_2(q_{\Psi})
\overline{q_{\frac{d\Phi_{\omega}}{d\omega}}}) -2 \| \Psi \|_{L^2}^2
\Re(JN_2(q_{\Psi^*}, q_{\Psi})
\overline{q_{\frac{d\Phi_{\omega}}{d\omega}}}) \right],\\

Z'_{10} =& 2 \frac{J N_2(q_{\Psi^*}, q_{\Psi})}{\overline{\kappa}},
\quad Z'_{01} = 2 \frac{J N_2(q_{\Psi})}{\overline{\kappa}}.
\end{array}
\end{equation}

\subsubsection{Equation for $h$}
In order to expand asymptotically the equation \eqref{eq-h} for $h$,
the following remark will be useful.
\begin{rmk}
For any $f \in L^2(\R^3)$ the following holds
$$P^c_T P^c f = P^c_T (I -P^d) f = P^c_T (P^c_T +P^d_T -P^d) f = P^c_T f +P^c_T
(P^d_T -P^d) f.$$

\end{rmk}

\noindent Let us denote
\begin{equation} \label{ro}
\rho(t) = \omega(t) -\omega_T +\dot{\gamma}(t),
\end{equation}

\noindent then equation \eqref{eq-h} can be rewritten as
$$\left( \frac{dh}{dt}, v \right)_{L^2} = Q_{L_T}(h,v) +(\rho P^c_T Jh,
v)_{L^2} +(8\pi P^c_T JN_2(q_\psi) G_{\lambda}, q_v
G_{\lambda})_{L^2} +$$
$$+\left( P^c_T \left[ \dot{\omega} \frac{dP^c}{d\omega} \psi +\dot{\gamma} P^c
J\psi +\rho Jg +\dot{\gamma} (P^d_T -P^d) Jf \right], v
\right)_{L^2} +$$
$$+\left( 8\pi \sqrt{\lambda} P^c_T \left( \frac{\sqrt{\omega}
-\sqrt{\omega_T}}{4\pi}\mathbb{T} q_f +P^c JN(q_{\chi})
-JN_2(q_\psi) \right) G_{\lambda}, q_v G_{\lambda} \right)_{L^2},$$

\noindent for any $v \in V$.

\noindent Denote
$$H'_R = P^c_T \left[ \dot{\omega} \frac{dP^c}{d\omega} \psi +\dot{\gamma} P^c
J\psi +\rho Jg +\dot{\gamma} (P^d_T -P^d) Jf \right],$$

\noindent and
$$H''_R = 8\pi \sqrt{\lambda} P^c_T \left( \frac{\sqrt{\omega}
-\sqrt{\omega_T}}{4\pi}\mathbb{T} q_f +P^c JN(q_{\chi})
-JN_2(q_\psi) \right) G_{\lambda}.$$


We recall that by $\Pi^{\pm}$ we denote (see Appendix C) the
projections onto the branches $\mathcal{C}_{\pm}$ of the continuous
spectrum separately. Analogously we denote by $\Pi^{\pm}_T$ the corresponding projections of the linear generator $L_T$ frozen
at time $T$. The following lemma is useful in the following.

\begin{lemma}   \label{lemma-LM}
There exists a constant $C > 0$ such that for each $h \in X^c_T$
holds

$$\left\| [P^c_T J -i(\Pi^+_T -\Pi^-_T)]h  \right\|_{L^1_w} \leq C \| h
\|_{L^{\infty}_{w^{-1}}}.$$

\end{lemma}

\noindent The proof is in Appendix \ref{dim-lemma} for any $t > 0$.
Finally, let us define

\begin{equation}    \label{L_M}
L_M(t) = L_T +i\rho(t) (\Pi^+_T -\Pi^-_T),
\end{equation}

\noindent then the previous equation becomes
\begin{equation}    \label{exp-h}
\left( \frac{dh}{dt}, v \right)_{L^2} = Q_{L_M}(h,v) +(8\pi P^c_T
JN_2(q_\psi) G_{\lambda}, q_v G_{\lambda})_{L^2} +(\widetilde{H}_R,
v)_{L^2} +(H''_R, q_v G_{\lambda})_{L^2},
\end{equation}

\noindent for any $v \in V$, where we denoted
$$\widetilde{H}_R = H'_R +\rho [P^c_T J -i(\Pi^+_T -\Pi^-_T)]h.$$

\noindent Finally, let us expand the second summand in the right
hand side of \eqref{exp-h}, and get
\begin{equation} \label{sv-h}
\left( \frac{dh}{dt}, v \right)_{L^2} = Q_{L_M}(h,v) +(z^2 H_{20}
+z\overline{z} H_{11} +\overline{z}^2 H_{02}) \overline{q_v}
+(\widetilde{H}_R, v)_{L^2} +(H''_R, q_v G_{\lambda})_{L^2},
\end{equation}
\noindent for any $v \in V$, where
$$\begin{array}{lll}
  H_{20} = (8\pi \sqrt{\lambda} P^c_T JN_2(q_{\Psi}) G_{\lambda},
G_{\lambda})_{L^2},\\
  H_{11} = 2 (8\pi \sqrt{\lambda} P^c_T JN_2(q_{\Psi}, q_{\Psi^*}) G_{\lambda},
G_{\lambda})_{L^2},\\
  H_{02} = (8\pi \sqrt{\lambda} P^c_T JN_2(q_{\Psi^*}) G_{\lambda},
G_{\lambda})_{L^2}.
\end{array}$$

\noindent Thanks to the estimates done for the other equations and
Lemma \ref{lemma-LM}, one can estimate the remainders in the
following way:
$$\|H'_R\|_{L^1_w} \leq C \left( |z| (|\dot{\omega}| +|\dot{\gamma}|)
+\mathcal{R}_1(\omega) (|\omega -\omega_T| +|\dot{\gamma}|
\|f\|_{L^{\infty}_{w^{-1}}}) \right) \leq$$
$$\leq \mathcal{R}_1(\omega, |z| +\|f\|_{L^{\infty}_{w^{-1}}}) \left( |z|^3 +|z|
\|f\|_{L^{\infty}_{w^{-1}}} +\|f\|_{L^{\infty}_{w^{-1}}}^2 +|\omega
-\omega_T| \|f\|_{L^{\infty}_{w^{-1}}} \right),$$

\noindent hence
\begin{equation}    \label{stime-H1}
\|\mathcal{\widetilde{H}}_R\|_{L^1_w} \leq \mathcal{R}_1(\omega, |z|
+\|f\|_{L^{\infty}_{w^{-1}}}) \left( |z|^3 +|z|
\|f\|_{L^{\infty}_{w^{-1}}} +\|f\|_{L^{\infty}_{w^{-1}}}^2 +|\omega
-\omega_T| \|f\|_{L^{\infty}_{w^{-1}}} \right),
\end{equation}

\noindent and
\begin{equation}    \label{stime-H2}
\|H''_R\|_{L^1_w} \leq \mathcal{R}_1(\omega, |z|
+\|f\|_{L^{\infty}_{w^{-1}}}) \left( |z|^3 +|z|
\|f\|_{L^{\infty}_{w^{-1}}} +\|f\|_{L^{\infty}_{w^{-1}}}^2 +|\omega
-\omega_T| (|z|^2 +\|f\|_{L^{\infty}_{w^{-1}}}) \right).
\end{equation}

\begin{rmk}
In the same way one directly expands the equation for the
function $f$ getting
\begin{equation}    \label{exp-f}
\left( \frac{df}{dt}, v \right)_{L^2} = Q_L(f,v)+(z^2 F_{20}
+z\overline{z} F_{11} +\overline{z}^2 F_{02}) \overline{q_v}
+(\widetilde{F}_R, v)_{L^2} +(F''_R, q_v G_{\lambda})_{L^2},
\end{equation}

\noindent for any $v \in V$, where
$$\begin{array}{lll}
  F_{20} = (8\pi \sqrt{\lambda} JN_2(q_{\Psi}) G_{\lambda},
G_{\lambda})_{L^2},\\
  F_{11} = 2 (8\pi \sqrt{\lambda} JN_2(q_{\Psi}, q_{\Psi^*}) G_{\lambda},
G_{\lambda})_{L^2},\\
  F_{02} = (8\pi \sqrt{\lambda} JN_2(q_{\Psi^*}) G_{\lambda},
G_{\lambda})_{L^2}.
\end{array}$$

\noindent and
$$\widetilde{F}_R = \dot{\omega} \frac{dP^c}{d\omega} \psi
+\dot{\gamma} P^c J\psi +\dot{\gamma} (P^d_T -P^d) Jf,$$
$$F''_R = 8\pi \sqrt{\lambda} \left( \frac{\sqrt{\omega}
-\sqrt{\omega_T}}{4\pi}\mathbb{T} q_f +P^c JN(q_{\chi})
-JN_2(q_\psi) \right) G_{\lambda}.$$

\noindent Furthermore, the $L^1_w$ norms of the remainders
$\widetilde{F}_R$ and $F''_R$ can be estimated by the corresponding
norms of the remainders $\widetilde{H}_R$ and $H''_R$.
\end{rmk}


\section{Canonical form of the equations}
In this section we would like to use the technique of normal
coordinates in order to transform the modulation equations for
$\omega$, $\gamma$, $z$, and $h$ to a simpler canonical form. We
will also try to keep the estimates of the remainders as much close
as possible to the original ones.

\subsection{Canonical form of the equation for $h$}
Our goal is to exploit a change of variable in such a way that the
function $h$ is mapped in a new function decaying in time at least
as $t^{-3/2}$. For this purpose one could expand $h$ as
\begin{equation}    \label{ansatz-h1}
h = h_1 +k +k_1,
\end{equation}

\noindent where
$$k = a_{20} z^2 +a_{11} z\overline{z} +a_{02} \overline{z}^2,$$

\noindent with some coefficients $a_{ij} = a_{ij}(x, \omega)$ such
that $a_{ij} = \overline{a_{ji}}$, and
$$k_1 = -\textrm{exp} \left( \int_0^t L_M(s) ds \right) k(0),$$
where the operator $L_M$ was defined in \eqref{L_M}.

\noindent Note that $h_1(0) = h(0)$, since $k_1(0) = -k(0)$.
\begin{prop}
There exist $a_{ij} \in L^{\infty}_{w^{-1}}(\R^3)$, for $i$, $j =
0$, $1$, $2$, such that the equation for $h_1$ has the form
\begin{equation}    \label{eq-h1}
\left( \frac{dh_1}{dt}, v \right)_{L^2} = Q_{L_M}(h_1, v)
+(\widehat{H}_R, v)_{L^2} +(H''_R, q_v G_{\lambda})_{L^2},
\end{equation}

\noindent for all $v \in V$, where $\widehat{H}_R = \widetilde{H}_R
+\overline{H}_R$ with
\begin{equation}    \label{stime-H3}
\overline{H}_R = -\left[ \dot{\omega} \left( \frac{d
a_{20}}{d\omega} z^2 +\frac{d a_{11}}{d\omega} z\overline{z}
+\frac{d a_{02}}{d\omega} \overline{z}^2 \right) +(2a_{20} z +a_{11}
\overline{z}) (\dot{z} -i\xi_T z) +\right.
\end{equation}
$$\left. +(a_{11} z +2a_{20} \overline{z}) (\dot{\overline{z}} +i\xi_T
\overline{z}) -\rho (\Pi^+_T -\Pi^-_T)k \right].$$

\end{prop}

\begin{proof}
The thesis is proved substituting \eqref{ansatz-h1} into
\eqref{exp-h} and equating the coefficients of the quadratic powers
of $z$ which leads to the system
\begin{equation}    \label{system-a}
\left\{ \begin{array}{ll}
           Q_{L_T}(a_{20},v) +\Re (H_{20} \overline{q_v}) -(2i\xi_T a_{20},
v)_{L^2}
= 0\\
           Q_{L_T}(a_{11},v) +\Re (H_{11} \overline{q_v}) = 0\\
           Q_{L_T}(a_{02},v) +\Re (H_{02} \overline{q_v}) +(2i\xi_T a_{02},
v)_{L^2} = 0
          \end{array} \right.,
\end{equation}

\noindent for all $v \in V$. The former system admits 
the solution
$$\begin{array}{ll}
a_{11} = -L_T^{-1} H_{11}\\

a_{20} = -(L_T -2i \xi_T -0)^{-1} H_{20}\\

a_{02} = \overline{a_{02}} = -(L_T +2i \xi_T -0)^{-1} H_{02}
\end{array}$$
\end{proof}

\begin{rmk}
From the explicit structure of the remainder $\widehat{H}_R$ it
follows that it still satisfies estimate \eqref{stime-H1}.
\end{rmk}

\noindent We will need to apply the next lemma which can be proved
as Proposition 2.3 in \cite{KKS}.
\begin{lemma}   \label{stime-evoluz}
If $\sigma \in \left( \frac{1}{\sqrt{2}}, \frac{\sqrt{3}
+1}{2\sqrt{2}} \right)$ and $f \in V \cap L^1_w$, then there exists
some constant $C > 0$ such that for any $t \geq 0$
$$\| e^{-L_T t} (L_T + 2i \xi_T -0)^{-1} P_T^c f
\|_{L^{\infty}_{w^{-1}}} \leq  C (1 +t)^{-3/2} \| f \|_{L^1_w}.$$

\end{lemma}

\begin{rmk}
Let us note that
$$h = P^c_Th = P^c_Th_1 +P^c_Tk +P^c_Tk_1,$$

\noindent hence, in order to estimate the decay of
$\|h\|_{L^{\infty}_{w^{-1}}}$, it suffices to estimate the decay of
$$\|P^c_Th_1\|_{L^{\infty}_{w^{-1}}}, \quad \|P^c_Tk\|_{L^{\infty}_{w^{-1}}},
\quad \textrm{and} \quad \|P^c_Tk_1\|_{L^{\infty}_{w^{-1}}}.$$

\end{rmk}

\subsection{Canonical form of the equation for $\omega$}
Since $\Omega_{11} = 0$, we can exploit the method by Buslaev and
Sulem in \cite{BS}, Proposition 4.1 and get the following
proposition.
\begin{prop} \label{zetuno}
There exist coefficients $b_{ij} = b_{ij}(\omega)$, with $i$, $j =
0$, $1$, $2$, $3$, and vector functions $b'_{ij} = b'_{ij}(x,
\omega)$, with $i$, $j = 0$, $1$, such that function
$$\omega_1 = \omega +b_{20} z^2 +b_{11} z\overline{z} +b_{02} \overline{z}^2
+b_{30} z^3 +b_{21} z^2\overline{z} +b_{12} z\overline{z}^2 +b_{03}
\overline{z}^3 +$$
$$+z (f, b'_{10})_{L^2} +\overline{z}(f, b'_{01})_{L^2},$$

\noindent solves a differential equation of the form
$$\dot{\omega}_1 = \widehat{\Omega}_R,$$

\noindent for some remainder $\widehat{\Omega}_R$.

\end{prop}

\begin{proof}
Substituting the equations \eqref{sv-omega}, \eqref{sv-z}, and
\eqref{exp-f} into the derivative with respect to time of the
expression for $\omega_1$ and equating the coefficients of $z^2$, $z
\overline{z}$, $\overline{z}^2$, $z$, and $\overline{z}$ one gets
the following system
$$\left\{ \begin{array}{ll}
       \Omega_{20} +2i \xi b_{20} = 0\\
           \Omega_{02} -2i \xi b_{02} = 0\\
       \Omega_{30} +3i \xi b_{30} +2Z_{20} b_{20} +\Re (F_{20}
\overline{q_{b'_{10}}}) = 0\\
       \Omega_{03} -3i \xi b_{03} +2Z_{02} b_{02} +\Re (F_{02}
\overline{q_{b'_{01}}}) = 0\\
       \Omega_{21} +i \xi b_{21} +2Z_{11} b_{20} +2Z_{20} b_{02} +\Re (F_{11}
\overline{q_{b'_{10}}} +F_{20} \overline{q_{b'_{01}}}) = 0\\
       \Omega_{12} -i \xi b_{12} +2Z_{11} b_{02} +2Z_{20} b_{20} +\Re (F_{11}
\overline{q_{b'_{01}}} +F_{20} \overline{q_{b'_{10}}}) = 0\\
       (q_f, \Omega'_{10}) +i\xi (f, b'_{10})_{L^2} +Q_L(f, b'_{10}) =0\\
       (q_f, \Omega'_{01}) +i\xi (f, b'_{01})_{L^2} +Q_L(f, b'_{01}) =0
          \end{array} \right..$$

\noindent The last two equations of this system can be solved in a
way similar to the ones system \eqref{system-a}, and the proof
follows.
\end{proof}

\begin{rmk}
From the proof of the previous proposition it also follows that the
remainder $\widehat{\Omega}_R$ can be estimated as $\Omega_R$,
namely
$$|\widehat{\Omega}_R| \leq \mathcal{R}(\omega, |z|
+\|f\|_{L^{\infty}_{w^{-1}}}) (|z|^2
+\|f\|_{L^{\infty}_{w^{-1}}})^2.$$

\end{rmk}

\noindent In the next lemma we prove a uniform bound for $|\omega_T
-\omega|$ on the interval $[0, T]$. For later convenience let us
denote
$$\mathcal{R}_2 (\omega, |z|+\|f\|_{L^{\infty}_{w^{-1}}}) = \mathcal{R}
\left(\max_{0 \leq t \leq T} |\omega_T -\omega|, \max_{0 \leq t \leq
T} (|z|+\|f\|_{L^{\infty}_{w^{-1}}}) \right).$$

\begin{rmk}
Let us note that $|\omega| \leq |\omega_0| +|\omega_0 -\omega_T|
+|\omega -\omega_T|$, then
$$\max_{0 \leq t \leq T} \mathcal{R} (\omega, |z|+\|f\|_{L^{\infty}_{w^{-1}}}) =
\mathcal{R} \left(\max_{0 \leq t \leq T} |\omega_T -\omega|, \max_{0
\leq t \leq T} (|z|+\|f\|_{L^{\infty}_{w^{-1}}}) \right).$$

\end{rmk}

\noindent The next lemma can be proved as in Section 3.5 of
\cite{KKS}.

\begin{lemma}
For any $t \in [0, T]$ we have
$$|\omega_T -\omega| \leq \mathcal{R}_2 (\omega,
|z|+\|f\|_{L^{\infty}_{w^{-1}}}) \left[ \int_t^T
(|z(\tau)|+\|f(\tau)\|_{L^{\infty}_{w^{-1}}})^2 d\tau +\right.$$
$$\left. +(|z_T|+\|f_T\|_{L^{\infty}_{w^{-1}}})^2
+(|z|+\|f\|_{L^{\infty}_{w^{-1}}})^2 \right].$$

\end{lemma}



\subsection{Canonical form of the equation for $\gamma$}
Equations \eqref{sv-gamma} for $\gamma$ and \eqref{sv-omega} for
$\omega$ differ just because in general $\Gamma_{11} \neq 0$. But we
can perform the same change of variable in the previous subsection,
namely
$$\gamma_1 = \gamma +d_{20} z^2 +d_{02} \overline{z}^2
+d_{30} z^3 +d_{21} z^2\overline{z} +d_{12} z\overline{z}^2 +d_{03}
\overline{z}^3 +z (f, d'_{10})_{L^2} +\overline{z}(f,
d'_{01})_{L^2},$$

\noindent for some suitable coefficients $d_{ij} = d_{ij}(\omega)$,
with $i$, $j = 0$, $1$, $2$, $3$, and vector functions $d'_{ij} =
d'_{ij}(x, \omega)$, with $i$, $j = 0$, $1$. Then the function
$\gamma_1$ solves the differential equation
$$\dot{\gamma}_1 = \Gamma_{11}(\omega) z \overline{z} +\widehat{\Gamma}_R,$$

\noindent for some remainder $\widehat{\Gamma}_R$, which can be
estimated as $\Gamma_R$, i.e.
$$|\widehat{\Gamma}_R| \leq \mathcal{R}(\omega, |z|
+\|f\|_{L^{\infty}_{w^{-1}}}) (|z|^2
+\|f\|_{L^{\infty}_{w^{-1}}})^2.$$

\subsection{Canonical form of the equation for $z$} \label{canonical-z}
Exploiting the change of variable \eqref{ansatz-h1} used to obtain
the canonical form of equation \eqref{exp-h} for $h$, one can prove
the following proposition.
\begin{prop}
There exist coefficients $c_{ij} = c_{ij}(\omega)$, with $i$, $j =
0$, $1$, $2$, $3$, such that function
$$z_1 = z +c_{20} z^2 +c_{11} z\overline{z} +c_{02} \overline{z}^2
+c_{30} z^3 +c_{21} z^2\overline{z} +c_{03} \overline{z}^3,$$

\noindent solves a differential equation of the form
\begin{equation}\label{zeta1}
\dot{z_1} = i\xi z_1 +i K |z_1|^2 z_1 +\widehat{Z}_R,
\end{equation}
\noindent where
$$iK = Z_{21} +Z'_{21} +\frac{i}{\xi} Z_{20} Z_{11} -\frac{i}{\xi}
Z_{11}^2 -\frac{2i}{3\xi} Z_{02}^2,$$

\noindent with the coefficient $Z_{ij}$, $i$, $j = 0, 1, 3$, defined
in \eqref{eqZij}, and
$$\widetilde{Z}_R = (g +P^c_Th_1 +P^c_Tk_1, Z'_{10}) z +(g +P^c_Th_1 +P^c_Tk_1,
Z'_{01}) \overline{z} +Z_R.$$

\end{prop}

\noindent The proof is a matter of calculation, but we give it explicitly to
stress the role of the functions $a_{ij}$, $i$, $j = 0$, $1$, $2$.

\begin{proof}
Substituting \eqref{ansatz-h1} in the equation \eqref{sv-z} the
differential equation for $z$ becomes
\begin{equation}    \label{eq-z-h1}
\dot{z} = i\xi z +Z_{20}z^2 +Z_{11}z\overline{z}+
Z_{02}\overline{z}^2 +Z_{30}z^3 +Z_{21}z^2 \overline{z}
+Z_{12}z\overline{z}^2 +Z_{03}\overline{z}^3 +
\end{equation}
$$+Z'_{30}z^3 +Z'_{21}z^2 \overline{z} +Z'_{12}z\overline{z}^2
+Z'_{03}\overline{z}^3 +\widetilde{Z}_R,$$

\noindent where
$$Z'_{30} = \Re (q_{a_{20}} \overline{Z'_{10}}),$$
$$Z'_{03} = \Re (q_{a_{02}} \overline{Z'_{01}}),$$
$$Z'_{21} = \Re (q_{a_{11}} \overline{Z'_{10}}) +\Re (q_{a_{20}} \overline{Z'_{01}}),$$
$$Z'_{12} = \Re (q_{a_{11}} \overline{Z'_{01}}) +\Re (q_{a_{02}} \overline{Z'_{10}}),$$

\noindent and the remainder $\widetilde{Z}_R$ is as in the statement
of the proposition.

\noindent Inserting equation \eqref{eq-z-h1} into the time
derivative of the expression for $z_1$ and equating the coefficients
of $z^2$, $z \overline{z}$, $\overline{z}^2$, $z^3$, $z
\overline{z}^2$, and $\overline{z}^3$ one obtains the system
$$\left\{ \begin{array}{ll}
  i\xi c_{20} +Z_{20} = 0\\
  -i\xi c_{11} +Z_{11} = 0\\
  -3i\xi c_{02} +Z_{02} = 0\\
  2i\xi c_{30} +Z_{30} +Z'_{30} +2c_{20} Z_{20} +c_{11} Z_{20} = 0\\
  Z_{12} +Z'_{12} +2c_{20} Z_{20} +c_{11} (Z_{11} +Z_{02}) +2c_{02}
  Z_{11} -2i\xi c_{12} = 0\\
  -4i\xi c_{03} +Z_{03} +Z'_{03} +c_{11} Z_{02} = 0
\end{array} \right.$$

\noindent The theorem follows from the fact the the above system is
solvable and in particular
$$c_{20} = \frac{i}{\xi} Z_{20}, \quad c_{11} = -\frac{i}{\xi}
Z_{11}, \quad \textrm{and} \quad c_{02} = \frac{i}{3\xi} Z_{02}.$$

\end{proof}

\begin{rmk}
For later convenience let us note that, since $Z_{21}$, $Z_{20}$,
$Z_{11}$, and $Z_{02}$ are purely imaginary, one has
$$\Re(iK) = \Re(Z'_{21}) 
.$$

\noindent Moreover, 
we need the following lemma.
\begin{lemma}   \label{reale}
There exists $\sigma^* \in \left( \frac{1}{\sqrt{2}}, \frac{\sqrt{3} +1}{2
\sqrt{2}}
\right]$ such that if $\sigma \in \left( \frac{1}{\sqrt{2}}, \sigma^* \right)$,
then
$$\Re(Z'_{21}) < 0,$$

\noindent $\forall \omega $ belonging to an open neighbourhood of $\omega_0$.
\end{lemma}

\begin{proof}
First of all recall that $\xi_T = 2 \sigma \sqrt{1 -\sigma^2} \omega_T$, then
one can compute
\begin{equation}    \label{kappa}
\kappa = -(\Psi, J \Psi)_{L^2} = \frac{i}{4\pi \sqrt{\omega_T}}
\left( \frac{1}{\sqrt{1 -2 \sigma \sqrt{1 -\sigma^2}}}
-\frac{(\sqrt{1 -\sigma^2} -1)^2}{\sigma^2} \frac{1}{\sqrt{1 +2
\sigma \sqrt{1 -\sigma^2}}} \right) =
\end{equation}
$$= \frac{i}{4\pi \sqrt{\omega_T}} \frac{\sigma^2 \sqrt{1 +2\sigma
\sqrt{1 -\sigma^2}} -(\sqrt{1-\sigma^2} -1)^2
\sqrt{1-2\sigma\sqrt{1-\sigma^2}}}{\sigma^2 (2\sigma^2 -1)}.$$

\noindent Since $\kappa$ is purely imaginary with
positive imaginary part and $L_T^{-1} 2 P_T^c J$ is self-adjoint, for the first
summand in the expression for $\Re(Z'_{21})$ one gets
$$\Re (q_{a_{11}} \overline{Z'_{10}}) = -2 \Re \left( \frac{q_{L_T^{-1} 2 P_T^c
J N_2(q_{\Psi}, q_{\Psi^*})} \overline{J N_2(q_{\Psi}, q_{\Psi^*})}}{\kappa}
\right) = 0.$$

\noindent Hence,
$$\Re(Z'_{21}) = -2 \Re \left( \frac{q_{a_{20}}
\overline{J N_2(q_{\Psi})}}{\kappa} \right).$$

\noindent By direct computations one has
$$a_{20}(x) = (L_T -(2i \xi_T +0))^{-1} H_{20} =  A \frac{e^{-\sqrt{\omega_T
+2\xi_T} |x|}}{4\pi |x|} \left( \begin{array}{cc}
1\\
-i
\end{array} \right) +C \frac{e^{-i\sqrt{-\omega_T +2\xi_T} |x|}}{4\pi |x|}
\left( \begin{array}{cc}
1\\
i
\end{array} \right),$$

\noindent with
$$\begin{array}{ll}
A = -\frac{4\pi}{d} [((2\sigma +1) \sqrt{\omega_T} -i \sqrt{-\omega_T +2\xi_T})
(H_{20})_1 + (i \sqrt{\omega_T} +\sqrt{-\omega_T +2\xi_T}) (H_{20})_2]\\

C = \frac{4\pi}{d} [((2\sigma +1) \sqrt{\omega_T} -\sqrt{\omega_T +2\xi_T})
(H_{20})_1 -(i \sqrt{\omega_T} -i \sqrt{\omega_T +2\xi_T}) (H_{20})_2]
\end{array},$$

\noindent where $d = 2i (2\sigma +1) \omega_T +2(\sigma +1)
\sqrt{\omega_T} \sqrt{-\omega_T +2\xi_T} -2i (\sigma +1)
\sqrt{\omega_T} \sqrt{\omega_T +2\xi_T} -2 \sqrt{\omega_T +2\xi_T}
\sqrt{-\omega_T +2\xi_T}$. From which follows
$$q_{a_{20}} = \frac{4\pi}{d} \left[ \left( \begin{array}{cc}
(i \sqrt{-\omega_T +2\xi_T} -\sqrt{\omega_T +2\xi_T}) (H_{20})_1\\

((2\sigma +1) \sqrt{\omega_T} -\sqrt{\omega_T +2\xi_T} -\sqrt{-\omega_T
+2\xi_T}) (H_{20})_1
\end{array} \right)  +\right.$$
$$\left. +\left( \begin{array}{cc}
-i (2\sqrt{\omega_T} +\sqrt{\omega_T +2\xi_T} +i \sqrt{-\omega_T
+2\xi_T}) (H_{20})_2\\
(- \sqrt{\omega_T +2\xi_T} +i \sqrt{-\omega_T +2\xi_T}) (H_{20})_2
\end{array} \right) \right].$$

\noindent Hence
\begin{equation}    \label{a20ReIm}
\begin{array}{ll}
\Re((q_{a_{20}})_1) = & \frac{16\pi}{|d|^2} [i(H_{20})_1 (-(\sigma
+1) \sqrt{\omega_T} \sqrt{\omega_T +2\xi_T} \sqrt{-\omega_T +2\xi_T}
+((\sigma +1) \omega_T +\xi_T) \sqrt{-\omega_T +2\xi_T})
+\\
& (H_{20})_2 (-(2(\sigma +1)\xi_T +(2\sigma +1) \omega_T) \sqrt{\omega_T} +(\xi_T +(2\sigma +1) \omega_T) \sqrt{\omega_T +2\xi_T})]\\
\Im((q_{a_{20}})_2) = & \frac{16\pi}{|d|^2} [i(H_{20})_1 ((2(\sigma +1)^2 \omega_T +\xi_T) \sqrt{-\omega_T +2\xi_T}
-(3\sigma +2) \sqrt{\omega_T} \sqrt{\omega_T +2\xi_T} \sqrt{-\omega_T +2\xi_T}) +\\
& +(H_{20})_2 ((\sigma +1) \omega_T^{3/2} +((\sigma +1) \omega_T
-\xi_T) \sqrt{\omega_T +2\xi_T})].
\end{array}
\end{equation}

\noindent Moreover, by \eqref{N} one gets
\begin{equation}    \label{N2}
 J N_2(q_{\Psi}) = \left( \begin{array}{cc}
-2\sigma |q_{\omega_T}|^{2\sigma -1} (q_{\Psi})_1 (q_{\Psi})_2\\

\sigma |q_{\omega_T}|^{2\sigma -1} (3 (q_{\Psi})_1^2 +(q_{\Psi})_2^2)+2 \sigma
(\sigma -1) |q_{\omega_T}|^{2\sigma-1} (q_{\Psi})_1^2
\end{array} \right) =
\end{equation}
$$= \left( \begin{array}{cc}
-2 i \sigma |q_{\omega_T}|^{2\sigma -1} \left( 1 -\frac{\sqrt{1 -\sigma^2}
-1}{\sigma} \right) \left( 1 +\frac{\sqrt{1 -\sigma^2} -1}{\sigma} \right)\\

2 \sigma |q_{\omega_T}|^{2\sigma -1} \left( 1 -\frac{\sqrt{1 -\sigma^2}
-1}{\sigma} \right)
\end{array} \right),$$

\noindent which implies
\begin{equation}    \label{H20}
H_{20} = (8 \pi \sqrt{\omega_T} P^c_T JN_2(q_{\Psi}) G_{\omega_T},
G_{\omega_T})_{L^2} =
\end{equation}
$$= JN_2(q_{\Psi}) -\frac{(JN_2(q_{\Psi}))_1 |q_{\omega_T}|}{ 16\pi \Delta
\omega^{3/2}} \left( \frac{1}{\sigma} -1 \right) \left( \begin{array}{cc}
1\\
0
\end{array} \right)+$$
$$+ \frac{\sqrt{\omega_T}}{\kappa} \left( \begin{array}{cc}
-(JN_2(q_{\Psi}))_2 \left( \frac{1}{\sqrt{\omega_T -\xi_T} +\sqrt{\omega_T}}
-\frac{\sqrt{1 -\sigma^2} -1}{\sigma} \frac{1}{\sqrt{\omega_T +\xi_T}
+\sqrt{\omega_T}} \right)^2\\

(JN_2(q_{\Psi}))_1 \left( \frac{1}{\sqrt{\omega_T -\xi_T} +\sqrt{\omega_T}}
+\frac{\sqrt{1 -\sigma^2} -1}{\sigma} \frac{1}{\sqrt{\omega_T +\xi_T}
+\sqrt{\omega_T}} \right)^2
\end{array} \right).$$

\noindent Let us notice that \eqref{N2} and \eqref{H20} imply
$$\begin{array}{ll}
i(H_{20})_1 i(JN_2(q_{\Psi}))_1 = & -\frac{1}{2\sigma -1} (JN_2(q_{\Psi}))_1^2
+\\
& +\frac{\sqrt{\omega_T}}{4\pi i\kappa} \left( \frac{1}{\sqrt{\omega_T -\xi_T}
+\sqrt{\omega_T}} -\frac{\sqrt{1 -\sigma^2} -1}{\sigma} \frac{1}{\sqrt{\omega_T
+\xi_T} +\sqrt{\omega_T}} \right)^2 i(JN_2(q_{\Psi}))_1 (JN_2(q_{\Psi}))_2\\
(H_{20})_2 i(JN_2(q_{\Psi}))_1 = & (JN_2(q_{\Psi}))_2 i(JN_2(q_{\Psi}))_1 +\\
& -\frac{\sqrt{\omega_T}}{4\pi i\kappa} \left( \frac{1}{\sqrt{\omega_T -\xi_T}
+\sqrt{\omega_T}} +\frac{\sqrt{1 -\sigma^2} -1}{\sigma} \frac{1}{\sqrt{\omega_T
+\xi_T} +\sqrt{\omega_T}} \right)^2 (JN_2(q_{\Psi}))_1^2\\
i(H_{20})_1 (JN_2(q_{\Psi}))_2 = & \frac{1}{2\sigma -1} i(JN_2(q_{\Psi}))_1
(JN_2(q_{\Psi}))_2 +\\
& +\frac{\sqrt{\omega_T}}{4\pi i\kappa} \left( \frac{1}{\sqrt{\omega_T -\xi_T}
+\sqrt{\omega_T}} -\frac{\sqrt{1 -\sigma^2} -1}{\sigma} \frac{1}{\sqrt{\omega_T
+\xi_T} +\sqrt{\omega_T}} \right)^2 (JN_2(q_{\Psi}))_2^2\\
(H_{20})_2 (JN_2(q_{\Psi}))_2 = & (JN_2(q_{\Psi}))_2^2
+\\
& +\frac{\sqrt{\omega_T}}{4\pi i\kappa} \left( \frac{1}{\sqrt{\omega_T -\xi_T}
+\sqrt{\omega_T}} +\frac{\sqrt{1 -\sigma^2} -1}{\sigma} \frac{1}{\sqrt{\omega_T
+\xi_T} +\sqrt{\omega_T}} \right)^2 i(JN_2(q_{\Psi}))_1 (JN_2(q_{\Psi}))_2,
\end{array}$$

\noindent then by \eqref{kappa} and \eqref{a20ReIm} it follows
$$\Re(Z'_{21}) = -2 \Re \left( \frac{q_{a_{20}}
\overline{J N_2(q_{\Psi})}}{\kappa} \right) = \frac{2}{i\kappa}
(\Re((q_{a_{20}})_1) i(JN_2(q_{\Psi}))_1 +\Im((q_{a_{20}})_2)
(JN_2(q_{\Psi}))_2) =$$
$$= \frac{128 \pi \omega_T^{3/2} |q_{\omega_T}|^{4\sigma -2}}{i\kappa
|d|^2} \sigma^2 \left( 1 -\frac{\sqrt{1 -\sigma^2} -1}{\sigma}
\right)^2 f(\sigma),$$

\noindent with

$$ f(\sigma) =
\left( \left( 2 (1+\sigma)^2+2 \sigma \sqrt{1-\sigma^2}\right)
\sqrt{-1+4 \sigma \sqrt{1-\sigma^2}}\right. -(2+3 \sigma) \sqrt{-1+4
\sigma \sqrt{1-\sigma^2}} \sqrt{1+4 \sigma \sqrt{1-\sigma^2}}+$$
$$+ \left(1+\frac{-1+\sqrt{1-\sigma^2}}{\sigma}\right) \left.\left((-1-\sigma) \sqrt{-1+16 \sigma^2-16
\sigma^4}+\left(1+\sigma+2 \sigma \sqrt{1-\sigma^2}\right)
\sqrt{-1+4 \sigma \sqrt{1-\sigma^2}}\right)\right) \cdot$$
$$\cdot \left(\frac{1}{-1+2
\sigma}-\frac{\left(\frac{1}{1+\sqrt{1-2 \sigma
\sqrt{1-\sigma^2}}}-\frac{-1+\sqrt{1-\sigma^2}}{\sigma+\sigma
\sqrt{1+2 \sigma \sqrt{1-\sigma^2}}}\right)^2}{\frac{1}{\sqrt{1-2
\sigma
\sqrt{1-\sigma^2}}}-\frac{\left(-1+\sqrt{1-\sigma^2}\right)^2}{\sigma^2
\sqrt{1+2 \sigma \sqrt{1-\sigma^2}}}}\right)+$$
$$+\left(1+\sigma+\left(1+\sigma-2 \sigma \sqrt{1-\sigma^2}\right)
\sqrt{1+4 \sigma \sqrt{1-\sigma^2}}+\right.$$
$$+ \left(1+\frac{-1+\sqrt{1-\sigma^2}}{\sigma}\right)
\left.\left(-1-2 \sigma+\left(-4 \sigma-4 \sigma^2\right)
\sqrt{1-\sigma^2}+\left(1+2 \sigma+2 \sigma \sqrt{1-\sigma^2}\right)
\sqrt{1+4 \sigma \sqrt{1-\sigma^2}}\right)\right) \cdot$$
$$\cdot \left(1-\frac{\left(\frac{1}{1+\sqrt{1-2 \sigma
\sqrt{1-\sigma^2}}}+\frac{-1+\sqrt{1-\sigma^2}}{\sigma+\sigma
\sqrt{1+2 \sigma \sqrt{1-\sigma^2}}}\right)^2}{\frac{1}{\sqrt{1-2
\sigma
\sqrt{1-\sigma^2}}}-\frac{\left(-1+\sqrt{1-\sigma^2}\right)^2}{\sigma^2
\sqrt{1+2 \sigma \sqrt{1-\sigma^2}}}}\right).$$

\noindent Notice that one has $f(\sigma) \rightarrow \widetilde{f} > 0$, $d
\rightarrow \widetilde{d}\neq 0$, and $i\kappa \rightarrow -\infty$ as
$\sigma \rightarrow 1/\sqrt{2}$; this implies
$$\lim_{\sigma \rightarrow 1/\sqrt{2}} \Re(Z_{21}') = \frac{128
\sqrt{2} \omega_T^{3/2} |q_{\omega_T}|^{2\sqrt{2}-2}}{\pi
|\widetilde{d}|^2} \widetilde{f} \lim_{\sigma \rightarrow
1/\sqrt{2}} \frac{1}{i \kappa} = 0^-.$$

\noindent Hence there is a neighborhood of $\frac{1}{\sqrt{2}}$
where $\Re(Z_{21}')$ is strictly negative. A {\it Mathematica} plot
of the function $f(\sigma)$ in the range $\left( \frac{1}{\sqrt{2}}
, \frac{\sqrt{3} +1}{2 \sqrt{2}} \right)$ is given in figure
\ref{graph}.

\begin{figure}[htbp]    \label{graph}
\begin{center}
\includegraphics[width=9cm]{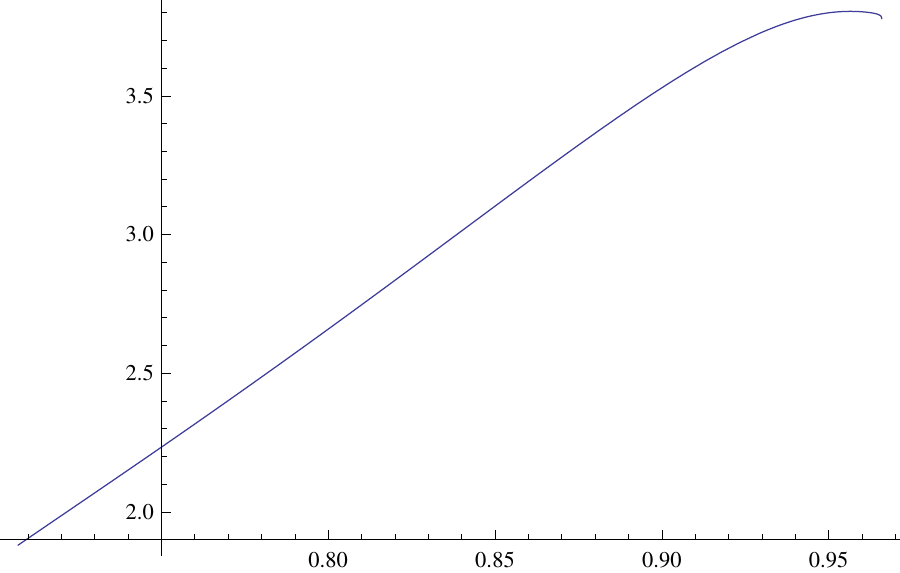}
\caption{$f(\sigma)$}
\end{center}
\end{figure}

\noindent Summing up, one can conclude that there exists $\sigma^* \in \left(
\frac{1}{\sqrt{2}}, \frac{\sqrt{3} +1}{2 \sqrt{2}} \right]$ such that
$\Re(Z_{21}') < 0$ for $\sigma \in \left( \frac{1}{\sqrt{2}}, \sigma^* \right)$.
\end{proof}

\end{rmk}

\begin{rmk}
\noindent The following reformulation of the equation for $z_1$ will
be useful. First of all, if we denote $K_T =
K(\omega_T)$, then the ordinary differential equation  for $z_1$
becomes
$$\dot{z_1} = i\xi z_1 +i K_T |z_1|^2 z_1 +\widehat{\widehat{Z}}_R,$$

\noindent for some remainder $\widehat{\widehat{Z}}_R$.

\noindent Secondly, let us notice that $z_1$ is oscillating while $y
= |z_1|^2$ decreases at infinity. Hence, it is easier to deal with
the variable $y$, which satisfies the equation
\begin{equation}    \label{eq-y}
\dot{y} = 2\Re(iK_T) y^2 +Y_R,
\end{equation}

\noindent where $Y_R$ is some suitable remainder.
\end{rmk}
\begin{rmk} \label{resto-z}
From Lemma \ref{lemma-cont} we have
$$|(g +P^c_Th_1 +P^c_Tk_1, Z'_{10})| \leq \mathcal{R}(\omega)
(\|g\|_{L^{\infty}_{w^{-1}}} +\|P^c_Th_1\|_{L^{\infty}_{w^{-1}}}
+\|P^c_Tk_1\|_{L^{\infty}_{w^{-1}}}) \leq$$
$$\leq \mathcal{R}_1(\omega) (|\omega_T -\omega| \|h\|_{L^{\infty}_{w^{-1}}}
+\|P^c_Th_1\|_{L^{\infty}_{w^{-1}}}
+\|P^c_Tk_1\|_{L^{\infty}_{w^{-1}}}),$$

\noindent hence
$$|Y_R| = |\widehat{\widehat{Z}}_R| |z| = |\widetilde{Z}_R +i(K -K_T) |z_1|^2
z_1| |z| \leq$$
$$\leq \mathcal{R}_1(\omega, |z| +\|f\|_{L^{\infty}_{w^{-1}}}) |z| [(|z|^2
+\|f\|_{L^{\infty}_{w^{-1}}})^2 +|z| |\omega_T -\omega| (|z|^2
+\|h\|_{L^{\infty}_{w^{-1}}}) +$$
$$+|z|\|P^c_Tk_1\|_{L^{\infty}_{w^{-1}}} +|z|
\|P^c_Th_1\|_{L^{\infty}_{w^{-1}}}].$$
\end{rmk}

\section{Majorants}\label{Majo}
In this section we exploit the so-called majorant method to prove
large time asymptotic for the solutions of the modulation equations.
Preliminarily, we need some assumptions on the initial conditions.

\subsection{Initial conditions}
Let us fix some $\epsilon > 0$ to be chosen later in order to obtain a
uniform control in
the estimates. Then we assume that

\begin{equation}    \label{initial-data}
\begin{array}{ll}
|z(0)| \leq \epsilon^{1/2}\\
\|f(0)\|_{L^1_w} \leq c \epsilon^{1/2},
\end{array}
\end{equation}

\noindent where $c > 0$ is some positive constant.

\noindent From the definition of $z_1$ (see Proposition \ref{zetuno}) 
one has
$$z_1 -z = \mathcal{R}(\omega) |z|^2.$$

\noindent Then the following estimate holds
$$y(0) = |z_1(0)|^2 \leq |z(0)|^2 +\mathcal{R}(\omega, |z(0)|) |z(0)|^3 \leq
\epsilon +\mathcal{R}(\omega, |z(0)|) \epsilon^{3/2}.$$

\noindent We also need an estimate for the initial datum of the
function $h(t)$. For this purpose recall that, from the definition of $f$ and $h$ one has the decomposition $h = f +(P^d
-P^d_T)f\ .$
 Hence,
$$\|h(0)\|_{L^1_w} \leq \|f(0)\|_{L^1_w} +\|(P^d -P^d_T)f(0)\|_{L^1_w} \leq
c\epsilon^{3/2} +\mathcal{R}_1(\omega) |\omega_T -\omega|
\|f(0)\|_{L^{\infty}_{w^{-1}}},$$

\noindent for some constant $c > 0$.

\noindent Thanks to the former estimates, one can prove the
following lemma.

\begin{lemma}
Let us assume conditions \eqref{initial-data} on the initial data.
Then
$$\|P^c_T k_1\|_{L^{\infty}_{w^{-1}}} \leq c \frac{|z(0)|^2}{(1+t)^{3/2}} \leq
\frac{c\epsilon}{(1+t)^{3/2}},$$

\noindent for all $t \geq 0$.
\end{lemma}

\begin{proof}
Let us denote $\zeta = \int_0^t \rho(\tau) d\tau$ (the quantity $\rho$
was defined in \eqref{ro}).

\noindent From the definition of the exponential and the idempotency
of the projections one gets
$$e^{i\zeta \Pi^{\pm}_T} = \Pi^{\pm}_T e^{i\zeta} +\Pi^{\mp}_T +P^d_T.$$

\noindent Then it follows
$$e^{i\zeta(\Pi^+_T -\Pi^-_T)} = (\Pi^+_T e^{i\zeta} +\Pi^-_T +P^d_T)
(\Pi^-_T e^{-i\zeta} +\Pi^+_T +P^d_T) = \Pi^+_T e^{i\zeta} +\Pi^-_T
e^{-i\zeta} +P^d_T.$$

\noindent The lemma follows from the fact that $L_T$ commutes with
the projectors $\Pi^{\pm}_T$, the definition \eqref{L_M} of the
operator $L_M$ and the decay of the evolution of the functions
$P_T^c a_{ij}$, $i$, $j = 0$, $1$, $2$, stated in Lemma
\ref{stime-evoluz}, namely
$$\|P^c_T k_1\|_{L^{\infty}_{w^{-1}}} = \left\| e^{\int_0^t L_M(\tau) d\tau}
P^c_T k(0) \right\|_{L^{\infty}_{w^{-1}}} =$$
$$= \| e^{L_T t} P^c_T (e^{i\zeta} \Pi^+_T +e^{-i\zeta}
\Pi^-_T+P^d_T) (a_{20} z^2(0) +a_{11} z(0) \overline{z(0)} +a_{02}
\overline{z(0)}^2) \|_{L^{\infty}_{w^{-1}}} \leq$$
$$\leq c \frac{|z(0)|^2}{(1+t)^{3/2}} \leq \frac{c \epsilon}{(1+t)^{3/2}}.$$

\end{proof}

\subsection{Definition of the majorants}
We are now in the position to define the majorants:
\begin{eqnarray}
M_0(T) = \max_{0 \leq t \leq T} |\omega_T -\omega| \left(
\frac{\epsilon}{1
+\epsilon t} \right)^{-1}\\
M_1(T) = \max_{0 \leq t \leq T} |z(t)| \left( \frac{\epsilon}{1
+\epsilon t} \right)^{-1/2}\\
M_2(T) = \max_{0 \leq t \leq T} \|P^c_T
h_1(t)\|_{L^{\infty}_{w^{-1}}} \left( \frac{\epsilon}{1 +\epsilon t}
\right)^{-3/2}
\end{eqnarray}

\noindent We shall use the following vector notation
\begin{equation}
M = (M_0, M_1, M_2).
\end{equation}

\begin{rmk}
From the estimates on $g$, $k_1$ and the definitions of the
majorants follows
$$\|f\|_{L^{\infty}_{w^{-1}}} = \|g +P^c_Th_1 +P^c_Tk
+P^c_Tk_1\|_{L^{\infty}_{w^{-1}}} \leq$$ $$ \leq
\mathcal{R}_1(\omega) \left( |\omega_T -\omega| +|z|^2
+\frac{\epsilon}{(1+t)^{3/2}} \|P^c_T h_1\|_{L^{\infty}_{w^{-1}}}
\right) \leq$$
$$\leq \frac{\epsilon}{1 +\epsilon t} \mathcal{R}_1(\omega) (M_1^2
+\epsilon^{1/2} M_2).$$

\end{rmk}

\noindent From the assumptions \eqref{initial-data} on the initial
data one obtains
$$y(0) \leq \epsilon +\mathcal{R}(\epsilon^{1/2} M) \epsilon^{3/2} \leq \epsilon
(1 +\mathcal{R}(\epsilon^{1/2} M) \epsilon^{1/2}),$$
$$\|h(0)\|_{L^1_w} \leq c \epsilon^{3/2} \mathcal{R} (\epsilon^{1/2} M)
\epsilon^2 M_0 (1 +M_1^2 +\epsilon^{1/2} M_2).$$

\subsection{The equation for $y$}
We aim at studying the asymptotic behavior of the solution of equation
\eqref{eq-y} for the variable $y$ introduced in Remark 2.18. To do
that we need the following lemma which is the analogous of Lemma 4.1
in \cite{KKS}.

\begin{lemma}
The remainder $Y_R$ in equation \eqref{eq-y} satisfies the estimate
$$|Y_R| \leq \mathcal{R}(\epsilon^{1/2} M) \frac{\epsilon^{5/2}}{(1
+\epsilon t)^2 \sqrt{\epsilon t}} (1 +|M|)^5.$$

\end{lemma}



\noindent Hence, equation \eqref{eq-y} is of the form

\begin{equation}
\dot{y} = 2 \Re(i K_T) y^2 + Y_R,
\end{equation}

\noindent with
$$\begin{array}{lll}
\Re(i K_T) < 0,\\
y(0) \leq \epsilon y_0,\\
|Y_R| \leq \overline{Y} \frac{\epsilon^{5/2}}{(1 +\epsilon t)^2
\sqrt{\epsilon t}},
\end{array}$$

\noindent where $y_0$ and $\overline{Y} > 0$ are some constants.
Then we can apply Proposition 5.6 in \cite{BS} and get the next
lemma.

\begin{lemma}   \label{stima-y}
Assuming the initial condition and the source term of equation
\eqref{eq-y} as above, the solution $y(t)$ is bonded as follows for
any $t > 0$
$$\left| y(t) -\frac{y(0)}{1 +2 \Im(K_T) y_0 t} \right| \leq c \overline{Y}
\left( \frac{\epsilon}{1 +\epsilon t} \right)^{3/2},$$

\noindent where $c = c(y_0, \Im(K_T))$.
\end{lemma}

\subsection{The equation for $P^c_T h_1$}
As a first step let us estimate the remianders in the equation
\eqref{eq-h1} for $h_1$. This is done in the next two lemmas.

\begin{lemma}
The remainders $\widetilde{H}_R$ and $H''_R$ can be estimated as
$$\|P^c_T\widetilde{H}_R\|_{L^1_w} \leq \mathcal{R}(\epsilon^{1/2} M) \left(
\frac{\epsilon}{1 +\epsilon t} \right)^{3/2} ((1 +M_1)^3
+\epsilon^{1/2} (1 +|M|)^4),$$

\noindent and
$$\|P^c_TH''_R\|_{L^1_w} \leq \mathcal{R}(\epsilon^{1/2} M) \left(
\frac{\epsilon}{1 +\epsilon t} \right)^{3/2} ((1 +M_1)^3
+\epsilon^{1/2} (1 +|M|)^4).$$

\end{lemma}

\begin{proof}
From the estimate \eqref{stime-H1} on $\widetilde{H}_R$ one has
$$\|P^c_T\widetilde{H}_R\|_{L^1_w} \leq \mathcal{R}_2 (\omega, |z|
+\|f\|_{L^{\infty}_{w^{-1}}}) [|z|^3 +(|z| +|\omega_T -\omega|)
(|z|^2 +\|P^c_T k_1\|_{L^{\infty}_{w^{-1}}} +$$
$$+\|P^c_T h_1\|_{L^{\infty}_{w^{-1}}})
+(|z|^2 +\|P^c_T k_1\|_{L^{\infty}_{w^{-1}}} +\|P^c_T
h_1\|_{L^{\infty}_{w^{-1}}})^2] \leq$$
$$\leq \mathcal{R}(\epsilon^{1/2} M) \left[ \left( \frac{\epsilon}{1 +\epsilon
t} \right)^{3/2} M_1^3 +\right.$$
$$\left. +\left( \left( \frac{\epsilon}{1 +\epsilon t}
\right)^{1/2} M_1 +\frac{\epsilon}{1 +\epsilon t} M_0 \right) \left(
\frac{\epsilon}{1 +\epsilon t} M_1^2 +\frac{\epsilon}{(1 +t)^{3/2}}
+\left( \frac{\epsilon}{1 +\epsilon t} \right)^{3/2} M_2 \right)
+\right.$$
$$\left. +\left( \frac{\epsilon}{1 +\epsilon t} M_1^2 +\frac{\epsilon}{(1
+t)^{3/2}} +\left( \frac{\epsilon}{1 +\epsilon t} \right)^{3/2} M_2
\right) \right] \leq$$
$$\leq \mathcal{R}(\epsilon^{1/2} M) \left(
\frac{\epsilon}{1 +\epsilon t} \right)^{3/2} ((1 +M_1)^3
+\epsilon^{1/2} (1 +|M|)^4).$$

\noindent The bound for $H''_R$ follows in the same way from the
estimate \eqref{stime-H2}.
\end{proof}

\noindent In the next lemma we get a estimate the evolution under
the linear operator $L_T$ of the remainder $P^c_T \overline{H}_R$.

\begin{lemma}
For any $t$, $s \geq 0$ the following estimate holds
$$\| e^{L_T t} P^c_T \overline{H}_R(s) \|_{L^{\infty}_{w^{-1}}} (1+t)^{3/2} \leq
\mathcal{R}(\epsilon^{1/2} M) \left( \frac{\epsilon}{1 +\epsilon s}
\right)^{3/2} (M_1^3 +\epsilon^{1/2} (1 +|M|)^3).$$

\end{lemma}

\begin{proof}
From the analytic expression \eqref{stime-H3} of $\overline{H}_R$
and the estimates of the evolution of the functions $a_{20}$,
$a_{11}$, and $a_{02}$ stated in Lemma \ref{stime-evoluz}, one has
$$\| e^{L_T t} P^c_T \overline{H}_R(s) \|_{L^{\infty}_{w^{-1}}} (1+t)^{3/2}
\leq$$
$$\leq \mathcal{R}_2(\omega, |z| +\|f\|_{L^{\infty}_{w^{-1}}})
|z| [|z| |\omega_T -\omega| +(|z| +\|k_1\|_{L^{\infty}_{w^{-1}}}
+\|h_1\|_{L^{\infty}_{w^{-1}}})^2] \leq$$
$$\leq \mathcal{R}(\epsilon^{1/2} M) \left( \frac{\epsilon}{1 +\epsilon s}
\right)^{1/2} M_1 \left[ \left( \frac{\epsilon}{1 +\epsilon s}
\right)^{3/2} M_0 M_1 +\right.$$
$$\left. +\left( \left( \frac{\epsilon}{1 +\epsilon s} \right)^{1/2} M_1
+\frac{\epsilon}{(1 +s)^{3/2}} +\left( \frac{\epsilon}{1 +\epsilon
s} \right)^{3/2} M_2 \right)^2 \right] \leq$$
$$\leq \mathcal{R}(\epsilon^{1/2} M) \left( \frac{\epsilon}{1 +\epsilon s}
\right)^{3/2} (M_1^3 +\epsilon^{1/2} (1 +|M|)^3).$$

\end{proof}

\noindent From the two previous lemmas we can get the following
result.

\begin{lemma}   \label{stima-h1}
Let us consider the equation for $P^c_T h_1$
$$\left(\frac{d P^c_T h_1}{dt}, v \right)_{L^2} = Q_{L_M}(P^c_T h_1, v)
+(P^c_T \widehat{H}_R, v)_{L^2} +(P^c_T H''_R, q_v
G_{\lambda})_{L^2},$$

\noindent with initial condition and source terms satisfying
$$\| h_1(0) \|_{L^1_w} \leq \epsilon^{3/2} h_0,$$
$$\widehat{H}_R = \widetilde{H}_R +H''_R,$$
\noindent such that
$$\|P^c_T\widetilde{H}_R\|_{L^1_w} \leq \overline{H}_1 \left( \frac{\epsilon}{1
+\epsilon t} \right)^{3/2},$$
$$\|P^c_TH''_R\|_{L^1_w} \leq \overline{H}_2 \left( \frac{\epsilon}{1 +\epsilon
t} \right)^{3/2},$$
$$\| e^{L_T t} P^c_T \overline{H}_R(s) \|_{L^{\infty}_{w^{-1}}} (1+t)^{3/2} \leq
\overline{H}_3 \left( \frac{\epsilon}{1 +\epsilon s} \right)^{3/2}
(M_1^3 +\epsilon^{1/2} (1 +|M|)^3).$$

\noindent for some positive constant $h_0$, $\overline{H}_1$,
$\overline{H}_2$ and $\overline{H}_3$. Then its solution is bounded
as follows
$$\| P^c_T h_1 \|_{L^{\infty}_{w^{-1}}} \leq c \left( \frac{\epsilon}{1
+\epsilon t} \right)^{3/2} (h_0 +\overline{H}_1 +\overline{H}_2
+\overline{H}_3),$$

\noindent where $c = c(\omega_T) > 0$.
\end{lemma}

\begin{proof}
By the Duhamel representation 
one has
$$(P^c_T h_1, v)_{L^2} = \left( e^{\int_0^t L_M(\tau) d\tau} h_1(0) +\int_0^t
e^{\int_s^t L_M(\tau) d\tau} P^c_T \widehat{H}_R(s) ds, v
\right)_{L^2} +$$
$$+\left( \int_0^t e^{\int_s^t L_M(\tau) d\tau} P^c_T H''_R(s) ds, q_v
G_{\lambda} \right)_{L^2},$$

\noindent for all $v \in V$.

\noindent Then from the dispersive estimate in Theorem
\ref{stima-cont} and the estimates on the remainders proved above in
the duality paring defined by the inner product $L^2$, one has
$$\|P^c_T h_1\|_{L^{\infty}_{w^{-1}}} = \sup_{0 \neq v \in L^1_w}
\frac{(P^c_T h_1, v)_{L^2}}{\|v\|_{V \cap L^1_w}} \leq$$
$$\leq c(\omega_T) \left( \frac{1}{(1+t)^{3/2}} \|h_1(0)\|_{L^1_w} +\int_0^t
\frac{1}{(1+t-s)^{3/2}} (\|P^c_T \widetilde{H}_R(s)\|_{L^1_w}
+\|P^c_T H''_R(s)\|_{L^1_w}) ds +\right.$$
$$\left. +\int_0^t \| e^{L_T (t-s)} P^c_T \overline{H}_R(s)
\|_{L^{\infty}_{w^{-1}}} ds \right) \leq$$
$$\leq c(\omega_T) \left( \left( \frac{\epsilon}{1+\epsilon t} \right)^{3/2} h_0
+\int_0^t \frac{1}{(1+t-s)^{3/2}} \left( \frac{\epsilon}{1 +\epsilon
s} \right)^{3/2} ds (\overline{H}_1 +\overline{H}_2 +\overline{H}_3)
\right).$$

\noindent The lemma follows from the fact that
$$\int_0^t \frac{1}{(1+t-s)^{3/2}} \left( \frac{\epsilon}{1 +\epsilon s}
\right)^{3/2} ds \leq c \left( \frac{\epsilon}{1 +\epsilon t}
\right)^{3/2},$$

\noindent for some constant $c > 0$.
\end{proof}

\subsection{Uniform bounds for the majorants}
To prove that the majorants are uniformly bounded, the following
lemma will be useful.

\begin{lemma}   \label{maj-ineq}
For any $T > 0$ the majorants $M_0$, $M_1$, and $M_2$ satisfy the
following inequalities
$$\begin{array}{lll}
M_0(T) \leq \mathcal{R}(\epsilon^{1/2} M) [(1 +M_1)^4 +\epsilon (1 +|M|)^2],\\

(M_1(T))^2 \leq \mathcal{R}(\epsilon^{1/2} M) [1 +\epsilon^{1/2} (1 +|M|)^5],\\

M_2(T) \leq \mathcal{R}(\epsilon^{1/2} M) [(1 +M_1)^3 +\epsilon^{1/2} (1 +|M|)^4].\\
\end{array}$$

\end{lemma}

\begin{proof}
It follows form Lemma \ref{stima-y} and \ref{stima-h1} as Lemma 4.6
in \cite{KKS}, but we give the proof for sake of completeness.

\noindent \texttt{Step 1.} Let us begin noting that
$$|z|^2 +\|f\|_{L^{\infty}_{w^{-1}}} \leq \mathcal{R}_2(\omega, |z|
+\|f\|_{L^{\infty}_{w^{-1}}}) (|z|^2 +\|P^c_T
k_1\|_{L^{\infty}_{w^{-1}}} +\|P^c_T h_1\|_{L^{\infty}_{w^{-1}}})
\leq$$
$$\leq \mathcal{R}(\epsilon^{1/2} M) \left( \frac{\epsilon}{(1+t)^{3/2}}
+\frac{\epsilon}{1 +\epsilon t} M_1^2 +\left( \frac{\epsilon}{1
+\epsilon t} \right)^{3/2} M_2 \right) \leq$$
$$\leq \mathcal{R}(\epsilon^{1/2} M) \frac{\epsilon}{1 +\epsilon t} (1 +M_1^2
+\epsilon^{1/2} M_2).$$

\noindent Then by the definition of $M_0$ and the bound on
$|\omega_T -\omega|$:
$$M_0(T) \leq \max_{0 \leq t \leq T} \left[ \left( \frac{\epsilon}{1 +\epsilon
t} \right)^{-1} \mathcal{R}(\epsilon^{1/2} M) \left( \int_t^T \left(
\frac{\epsilon}{1 +\epsilon \tau} \right)^2 (1 +M_1(\tau)^2 +\right.
\right.$$
$$\left. \left. +\epsilon^{1/2} M_2(\tau))^2 d\tau +\left( \frac{\epsilon}{1
+\epsilon t} \right)^2 (1 +M_1^2 +\epsilon^{1/2} M_2)^2 \right)
\right] \leq$$
$$\leq \mathcal{R}(\epsilon^{1/2} M) [(1 +M_1)^4 +\epsilon (1 +|M|)^2].$$

\noindent \texttt{Step 2.} Since $y = |z_1|^2$, we can exploit the
inequality proved in Lemma \ref{stima-y}, the fact that
$\overline{Y} = \mathcal{R}(\epsilon^{1/2} M) (1 +|M|)^5$ and $y(0)
\leq \epsilon y_0$, one gets
$$y \leq \mathcal{R}(\epsilon^{1/2} M) \left[ \frac{\epsilon}{1 +\epsilon t}
+\left( \frac{\epsilon}{1 +\epsilon t} \right)^{3/2} (1 +|M|)^5
\right].$$

\noindent From which follows
$$|z|^2 \leq y +\mathcal{R}(\omega) |z|^3 \leq$$
$$\leq \mathcal{R}(\epsilon^{1/2} M) \left[ \frac{\epsilon}{1 +\epsilon
t} +\left( \frac{\epsilon}{1 +\epsilon t} \right)^{3/2} (1 +|M|)^5
+\left( \frac{\epsilon}{1 +\epsilon t} \right)^{3/2} M_1^3 \right]
\leq \mathcal{R}(\epsilon^{1/2} M) [1 +\epsilon^{1/2} (1 +|M|)^5].$$

\noindent \texttt{Step 3.} Recall that
$$\|h(0)\|_{L^1_w} \leq c \epsilon^{3/2} \mathcal{R} (\epsilon^{1/2} M)
\epsilon^2 M_0 (1 +M_1^2 +\epsilon^{1/2} M_2),$$
$$\overline{H}_1 = \mathcal{R}(\epsilon^{1/2} M) ((1 +M_1)^3 +\epsilon^{1/2} (1
+|M|)^4),$$
$$\overline{H}_2 = \mathcal{R}(\epsilon^{1/2} M) ((1 +M_1)^3 +\epsilon^{1/2} (1
+|M|)^4),$$
$$\overline{H}_3 = \mathcal{R}(\epsilon^{1/2} M) (M_1^3 +\epsilon^{1/2} (1
+|M|)^3).$$

\noindent Hence from Lemma \ref{stima-h1} follows
$$\|P^c_T h_1\|_{L^{\infty}_{w^{-1}}} \leq \mathcal{R}(\epsilon^{1/2} M) \left(
\frac{\epsilon}{1 +\epsilon t} \right)^{3/2} ((1 +M_1)^3
+\epsilon^{1/2} (1 +|M|)^4),$$

\noindent which implies the inequality for $M_2$.
\end{proof}

\noindent We are now in the position to prove the uniform
boundedness of the majorants.

\begin{prop}
If $\epsilon > 0$ is sufficiently small, there exist a positive
constant $\overline{M}$ independent of $T$ and $\epsilon$ such that
$$|M(T)| \leq \overline{M},$$

\noindent for all $T > 0$.
\end{prop}

\begin{proof}
From the previous lemma follows
$$|M|^2 \leq \mathcal{R}(\epsilon^{1/2} M) [(1 +M_1)^8 +\epsilon^{1/2} (1
+|M|)^8] \leq \mathcal{R}(\epsilon^{1/2} M) (1 +\epsilon^{1/2}
F(M)),$$

\noindent where in the last inequality we have replaced the estimate
for $M_1^2$, and $F(M)$ is a suitable polynomial function.

\noindent Furthermore, $M(0)$ is small and $M(T)$ is a continuous
function. Hence it follows that  $|M|$ is bounded independent of
$\epsilon \ll 1$.
\end{proof}

\noindent The last proposition gives a summary of the behavior of
the functions $\omega(t)$, $z(t)$, $P^c_T h_1(t)$, and $f(t)$.

\begin{cor} \label{limit-existence}
There exists a finite limit $\omega_{\infty}$ for the function
$\omega(t)$ as $t \rightarrow +\infty$. Moreover the following holds
for all $t > 0$
$$\begin{array}{llll}
|\omega_{\infty} -\omega(t)| \leq \overline{M} \frac{\epsilon}{1
+\epsilon t},\\
|z(t)| \leq \overline{M} \left( \frac{\epsilon}{1 +\epsilon t}
\right)^{1/2},\\
\|P^c_T h_1(t)\|_{L^{\infty}_{w^{-1}}} \leq \overline{M} \left(
\frac{\epsilon}{1 +\epsilon t} \right)^{3/2},\\
\|f(t)\|_{L^{\infty}_{w^{-1}}} \leq \overline{M} \frac{\epsilon}{1
+\epsilon
t}.\\
\end{array}$$

\end{cor}

\section{Large time behavior of the solution and scattering asymptotics}
\subsection{Large time behavior of the solution of equation \eqref{eq1}}

The results of the previous section lead to the following
theorem.

\begin{theo}    \label{teo-infinito}
Let $u(t)$ be a solution of equation \eqref{eq1} with initial datum
$u_0 \in V \cap L^1_w$ of the form
$$u_0(x) = e^{i \theta_0} \Phi_{\omega_0}(x) +z_0 \Psi(x) +\overline{z}_0
\Psi^*(x) +f_0(x),$$

\noindent where $\theta_0 \in \R$, $\omega_0 > 0$, $z_0 \in \C$ with
$$|z(0)| \leq \epsilon^{1/2}, \qquad \|f_0\|_{L^1_w} \leq c \epsilon^{3/2},$$

\noindent for some $\epsilon$, $c > 0$. Then, provided $\epsilon$ is
small enough, there exist $\omega(t)$, $\gamma(t)$, $z(t) \in
C^1([0, +\infty))$ solutions of the modulation equations
\eqref{mod-eq1}-\eqref{mod-eq3}, and two constants
$\omega_{\infty}$, $\overline{M} > 0$ such that $\displaystyle
\omega_{\infty} = \lim_{t \rightarrow +\infty} \omega(t)$ and for
all $t \geq 0$
$$u(t,x) = e^{i(\int_0^t \omega(s) ds +\gamma(t))} \left( \Phi_{\omega(t)}(x)
+z(t) \Psi(t,x) +\overline{z(t)} \Psi^*(t,x) +f(t,x) \right),$$

\noindent where
$$|\omega_{\infty} -\omega(t)| \leq \overline{M} \frac{\epsilon}{1 +\epsilon
t}, \quad |z(t)| \leq \overline{M} \left( \frac{\epsilon}{1
+\epsilon t} \right)^{1/2}, \quad \|f(t)\|_{L^{\infty}_{w^{-1}}}
\leq \overline{M} \frac{\epsilon}{1 +\epsilon t}.$$

\end{theo}

\begin{proof}
Let us recall that the decomposition of the function $f$ as
$$f = g +h_1 +k +k_1$$

\noindent depends on the quantity $\omega(T)$. On the other hand
Corollary \ref{limit-existence} claims that the function $\omega(t)$
converges to some $\omega_{\infty} > 0$ as $t \rightarrow +\infty$.

\noindent As a consequence, one can reformulate the decomposition by
choosing $T = +\infty$. Moreover, all the estimates obtained before
for finite $T$ can be extended to $T = +\infty$ without
modification. Hence the theorem.
\end{proof}

\noindent The next goal is to construct precise asymptotic
expressions for $\omega(t)$, $\gamma(t)$, and $z(t)$. For later
convenience let us define (recall that $\xi$ depends explicitly on
$\omega$, see \eqref{mu}; and similarly for $K$, see \eqref{zeta1}
and subsequent, and $\gamma$)
$$\xi_{\infty} = \xi(\omega_{\infty}),$$
$$\gamma_{\infty} = \gamma(\omega_{\infty}),$$
$$K_{\infty} = K(\omega_{\infty}).$$

\begin{lemma}   \label{large-time}
Under the assumption of Theorem \ref{teo-infinito} the functions
$\omega(t)$, $\gamma(t)$, and $z(t)$ have the following asymptotic
behavior as $t \rightarrow +\infty$:
$$\omega(t) = \omega_{\infty} + \frac{q_1}{1 +\epsilon k_{\infty} t}
+\frac{q_2}{1 +\epsilon k_{\infty} t} \cos(2 \xi_{\infty} t +a_1
\log(1 +\epsilon k_{\infty} t) +a_2) +O(t^{-3/2}),$$
$$\gamma(t) = \gamma_{\infty} +b_1 \log(1 +\epsilon k_{\infty} t) +O(t^{-1}),$$
$$z(t) = z_{\infty} \frac{e^{i\int_0^t \xi(\tau) d\tau}}{(1 +\epsilon k_{\infty}
t)^{\frac{1 -i\delta}{2}}} +O(t^{-1}),$$

\noindent where
$$z_{\infty} = z_1(0) +\int_0^{+\infty} e^{-i\int_0^s \xi(\tau) d\tau} (1
+\epsilon k_{\infty} s)^{\frac{1 -i\delta}{2}} Z_1(s) ds,$$

$\epsilon k_{\infty} = 2\Im(K_{\infty})y_0$, $\delta =
\frac{\Re(K_{\infty})}{\Im(K_{\infty})}$, and $q_1$, $q_2$, $a_1$,
$a_2$, $b_1$ are constants.
\end{lemma}

\begin{proof}
We will prove the asymptotics for $z(t)$ only; the formulas for
$\omega(t)$ and $\gamma(t)$ can be deduced as in Sections 6.1 and
6.2 of \cite{BS}.

\noindent In order to do that let us recall that the equation for
$z_1(t)$ can be written as
$$\dot{z_1} = i\xi z_1 +i K_{\infty} |z_1|^2 z_1 +\widehat{\widehat{Z}}_R,$$

\noindent moreover Remark \ref{resto-z} and the inequalities
satisfied by the majorants in Lemma \ref{maj-ineq} justify the
following estimates on $\widehat{\widehat{Z}}_R$
$$|\widehat{\widehat{Z}}_R| \leq \mathcal{R}_1(\omega, |z|
+\|f\|_{L^{\infty}_{w^{-1}}}) [(|z|^2
+\|f\|_{L^{\infty}_{w^{-1}}})^2 +|z| |\omega_T -\omega| (|z|^2
+\|h\|_{L^{\infty}_{w^{-1}}}) +$$
$$+|z|\|P^c_Tk_1\|_{L^{\infty}_{w^{-1}}} +|z|
\|P^c_Th_1\|_{L^{\infty}_{w^{-1}}}] \leq$$
$$\leq \mathcal{R}(\epsilon^{1/2} M) \frac{\epsilon^2}{(1 +\epsilon t)^{3/2}
\sqrt{\epsilon t}} (1 +\overline{M}^4) = O(t^{-2}),$$

\noindent as $t \rightarrow +\infty$. On the other hand, Lemma
\ref{stima-y} implies
$$y(t) = \frac{y(0)}{1 +2\Im(K_{\infty}) y(0) t} +O(t^{-3/2}), \qquad
\textrm{as} \; t \rightarrow +\infty.$$

\noindent Let us note that $|z_1|$ satisfies the same bound of
$|z|$, namely
$$|z_1| \leq \overline{M} \left( \frac{\epsilon}{1 +\epsilon t} \right)^{1/2},$$

\noindent then the equation for $z_1(t)$ can be rewritten in the
formulas
$$\dot{z_1} = i\xi z_1 +i K_{\infty} \frac{y(0)}{1 +2\Im(K_{\infty}) y(0) t} z_1
+Z_1,$$

\noindent where $Z_1 = O(t^{-2})$ as $t \rightarrow +\infty$.

\noindent Since $y(0) = \epsilon y_0$, one has $\epsilon K_{\infty}
y_0 = \frac{i}{2} \epsilon k_{\infty} (1 -i\delta)$ and the equation
for $z_1(t)$ becomes
$$\dot{z_1} = \left( i\xi -\frac{i}{2} \epsilon k_{\infty} (1 -i\delta)
\frac{1}{1 +\epsilon k_{\infty} t} \right) z_1 +Z_1.$$

\noindent Hence, one gets
$$z_1(t) = \frac{e^{i \int_0^t \xi(\tau) d\tau}}{(1 +\epsilon k_{\infty}
t)^{\frac{1 -i\delta}{2}}} \left( z_1(0) +\int_0^s e^{-i \int_0^t
\xi(\tau) d\tau} (1 +\epsilon k_{\infty} s)^{\frac{1 -i\delta}{2}}
ds \right) = z_{\infty} \frac{e^{i \int_0^t \xi(\tau) d\tau}}{(1
+\epsilon k_{\infty} t)^{\frac{1 -i\delta}{2}}} +z_R,$$

\noindent where $z_{\infty}$ is as in the statement of the lemma and
$$z_R = -\int_t^{+\infty} e^{i \int_s^t \xi(\tau) d\tau} \left( \frac{1
+\epsilon k_{\infty} s}{1 +\epsilon k_{\infty} t} \right)^{\frac{1
-i\delta}{2}} Z_1(s) ds.$$

\noindent The bound on $Z_1$ implies $z_R = O(t^{-1})$. Therefore
$z(t)$ has the asymptotic behavior as $t \rightarrow +\infty$ stated
in the lemma because
$$z(t) = z_1(t) +O(t^{-1}) =  z_{\infty} \frac{e^{i\int_0^t \xi(\tau) d\tau}}{(1
+\epsilon k_{\infty} t)^{\frac{1 -i\delta}{2}}} +O(t^{-1}).$$

\end{proof}

\subsection{Scattering asymptotics}
Let us make the following ansatz
$$u(t,x) = s(t,x) +\zeta(t,x) +f(t,x),$$

\noindent where
$$s(t,x) = e^{i\Theta(t)} \Phi_{\omega(t)}(x),$$

\noindent is the modulated soliton and
$$\zeta(t,x) = e^{i\Theta(t)} [(z(t) +\overline{z}(t)) \Psi_1(x) +i(z(t)
-\overline{z}(t)) \Psi_2(x)]$$

\noindent is the fluctuating component. Recall that the functions
$\Phi_{\omega}$, $\Psi_1$ and $\Psi_2$ satisfy
$$\omega \Phi_{\omega} = -H_{\alpha} \Phi_{\omega},$$
$$\omega \Psi_1 = -i\xi \Psi_2 -H_{\alpha_1}\Psi_1,$$
$$\omega \Psi_2 = i\xi \Psi_1 -H_{\alpha_2}\Psi_2.$$

\noindent Therefore from equation \eqref{eq1} one gets
$$\left( i \frac{df}{dt}, v \right)_{L^2} = Q_0(f, v) -\nu
(|q_u|^{2\sigma}q_u -|q_s|^{2\sigma}q_s -\alpha_1 q_{(z
+\overline{z})\Psi_1} -\alpha_2 q_{(z -\overline{z})\Psi_2})
\overline{q_v} +$$
$$+(\dot{\gamma} (s +\zeta) -i\dot{\omega} \frac{d}{d\omega}(s +\zeta) -i
e^{i\Theta} [(\dot{z} -i\xi z) (\Psi_1 +i\Psi_2)
+(\dot{\overline{z}} -i\xi \overline{z}) (\Psi_1 -i\Psi_2)],
v)_{L^2},$$

\noindent for all $v \in V$, where $Q_0$ is the quadratic form of
the free Laplacian. Hence, as in \cite{ADFT}, the solution $f(t)$
can be formally expressed as
$$f(t,x) = U_t *f_0(x) +i \int_0^t U_{t -\tau}(x) q_f(\tau) d\tau -i \int_0^t
U_{t -\tau} *G(\tau) d\tau,$$

\noindent where we have denoted
$$G(t) = \dot{\gamma}(t) (s(t) +\zeta(t)) -i\dot{\omega}(t)
\frac{d}{d\omega}(s(t) +\zeta(t)) +$$
$$-i e^{i\Theta(t)} [(\dot{z}(t) -i\xi z(t)) (\Psi_1(t) +i\Psi_2(t))
+(\dot{\overline{z}}(t) -i\xi \overline{z}(t)) (\Psi_1(t)
-i\Psi_2(t))]$$

\noindent and $U_t(x) = \frac{e^{i \frac{|x|^2}{4t}}}{(4\pi
it)^{3/2}}$ is the propagator of the free Laplacian in $\R^3$.

\noindent In order to prove the asymptotic stability result we need
the two following lemmas.

\begin{lemma}
If the assumptions of Theorem \ref{teo-infinito} hold true, then
$$\int_0^t U_{t -\tau}(x) q_f(\tau) d\tau = U_t * \int_0^{+\infty} U_{-\tau}(x)
q_f(\tau) d\tau -\int_t^{+\infty} U_{t -\tau}(x) q_f(\tau) d\tau =
U_t * \phi_0 +r_0,$$

\noindent where $\phi_0 \in L^2(\R^3)$ and $r_0 = O(t^{-1/4})$ as $t
\rightarrow +\infty$ in $L^2(\R^3)$.
\end{lemma}

\begin{proof}
The strategy is similar to the one exploited in the case without eigenvalues (see the proof of Theorem 7.1 in \cite{ADO}): since
$\phi_0(x) = \frac{1}{(4\pi i)^{3/2}} \widetilde{\phi}_0 \left(
\frac{|x|^2}{4} \right)$, for some function $\widetilde{\phi}_0: \R^+
\rightarrow \C$, one gets
$$\| \phi_0 \|^2_{L^2} = \frac{1}{(4\pi)^2} \int_0^{+\infty} \left|
\widetilde{\phi}_0 \left( \frac{r^2}{4} \right) \right|^2 r^2 dr =
\frac{1}{(2\pi)^2} \int_0^{+\infty} | \widetilde{\phi}_0 (y) |^2
\sqrt{y} dy.$$

\noindent Hence $\phi_0 \in L^2(\R^3)$ if and only if
$\widetilde{\phi}_0 \in L^2(\R^+, \sqrt{y} dy)$. On the other hand,
one can make the change of variables $u = \frac{1}{\tau}$ in the
integral that defines the function $\widetilde{\phi}_0$ and get
$$\widetilde{\phi}_0(y) = \int_0^{+\infty} e^{-iyu} \frac{1}{u} q_f \left(
\frac{1}{u} \right) \sqrt{u} du,$$

\noindent then $\widehat{\widetilde{\phi}_0} = \frac{1}{u} q_f
\left( \frac{1}{u} \right)$. Moreover, by Corollary
\ref{limit-existence} one has
$$\| \widehat{\widetilde{\phi}_0} \|_{L^2}^2 = \int_0^{+\infty} \frac{1}{u^2}
\left| q_f \left( \frac{1}{u} \right) \right|^2 \sqrt{u} du \leq C
\int_0^{+\infty} \frac{\sqrt{u}}{(u +\epsilon)^2} du \leq C,$$

\noindent for some constant $C > 0$, hence the Plancherel identity
implies
$$\widetilde{\phi}_0 \in L^2(\R^+, \sqrt{y} dy).$$

\noindent In the same way, for any $t > 0$ the following holds
$$\| r_0 \|_{L^2}^2 = \frac{1}{(2\pi)^2} \left\| \frac{1}{u} q_f \left( t
+\frac{1}{u} \right) \right\|_{L^2(\R^+, \sqrt{u} du)}^2 \leq C
\frac{1}{\sqrt{1 +\epsilon t}},$$

\noindent for some constant $C > 0$ independent of $t$. Which
concludes the proof.
\end{proof}

\noindent The analogous result for the integral function $\int_0^t
U_{t -\tau} * G(\tau) d\tau$ requires different tools.

\begin{lemma}
Assume that the assumptions of Theorem \ref{teo-infinito} hold true,
then
$$\int_0^t U_{t -\tau} * G(\tau) d\tau = U_t * \int_0^{+\infty} U_{-\tau} *
G(\tau) d\tau -U_t * \int_t^{+\infty} U_{-\tau} * G(\tau) d\tau =
U_t * \phi_1 +r_1,$$

\noindent where $\phi_1 \in L^2(\R^3)$ and $r_1 = O(t^{-1/2})$ as $t
\rightarrow +\infty$ in $L^2(\R^3)$.
\end{lemma}

\begin{proof}
We exploit the idea used in \cite{KKS} to prove Lemma 5.5.

\noindent \texttt{Step 1: restriction to the leading terms.}

\noindent From the expansions \eqref{sv-omega}, \eqref{sv-gamma} and
\eqref{sv-z} for $\dot{\omega}(t)$, $\dot{\gamma}(t)$ and
$\dot{z}(t) -i\xi z(t)$ follow that the function $G(t)$ is made by a
quadratic part consisting in the terms multiplied $e^{i\Theta (t)}
z_{\infty}^2$, $e^{i\Theta (t)} \overline{z_{\infty}}^2$ or
$e^{i\Theta (t)} |z_{\infty}|^2$, with
$$z_{\infty} = \frac{e^{i \xi_{\infty} t}}{\sqrt{1 +\epsilon k_{\infty} t}}, $$

\noindent which are of order $t^{-1}$ and a remainder of order
$t^{-3/2}$. The convergence and the decay of the remainder is
trivial from the unitarity of $U_t$. Furthermore, from the analytic
definition of $G$ it follows that it is a complex linear combination
of functions of the form
$$Q(x) = e^{-\sqrt{\alpha} |x|^2}, \qquad \alpha = \omega_{\infty},
\omega_{\infty} +\nu_{\infty}, \omega_{\infty} -\nu_{\infty}.$$

\noindent Hence it suffices to prove the lemma for the functions
$\Pi(t) Q(x)$, where $\Pi(t)$ is one between $e^{i\Theta (t)}
z_{\infty}^2$, $e^{i\Theta (t)} \overline{z_{\infty}}^2$ and
$e^{i\Theta (t)} |z_{\infty}|^2$.

\noindent \texttt{Step 2: decomposition of $U_t * Q$.}

\noindent Let us note that we can rewrite the convolution product as
follows
$$U_t * Q = \frac{e^{i\frac{|x|^2}{4t}}}{(4\pi it)^{3/2}} \int_{\R^3}
e^{-i\frac{(x,y)}{2t}} Q(y) dy +\frac{e^{i\frac{|x|^2}{4t}}}{(4\pi
it)^{3/2}} \int_{\R^3} e^{-i\frac{(x,y)}{2t}} (e^{i\frac{|y|^2}{4t}}
-1) Q(y) dy =$$
\begin{equation}    \label{decomp-conv}
= \frac{e^{i\frac{|x|^2}{4t}}}{(2it)^{3/2}} \widehat{Q}
\left(\frac{x}{2t} \right)
+\frac{e^{i\frac{|x|^2}{4t}}}{(2it)^{3/2}} \widehat{Q_t}
\left(\frac{x}{2t} \right),
\end{equation}

\noindent where $Q_t(y) = (e^{i\frac{|y|^2}{4t}} -1) Q(y)$.

\noindent Since $|e^{i\theta} -1| \leq \theta$ and the function
$G(y)$ is exponentially decaying as $|y| \rightarrow +\infty$, the
$L^2$ norm of the second term of \eqref{decomp-conv} can be
estimated in the following way for any $t > 1$
$$\frac{1}{(2t)^{3/2}} \left\| \widehat{Q_t} \left(\frac{\cdot}{2t} \right)
\right\|_{L^2} = \| \widehat{Q_t}(\cdot) \|_{L^2} \leq \frac{1}{4t}
\left( \int_{\R^3} |y|^4 |Q(y)|^2 dy \right)^{1/2} \leq
\frac{C}{t},$$

\noindent for some constant $C > 0$. Hence, recalling that
$\Pi(\tau) \leq (1 +\epsilon k_{\infty} \tau)^{-1}$, we obtain
$$\int_0^{+\infty} \Pi(\tau) U_{\tau} * Q_t d\tau \in L^2(\R^3),$$

\noindent and
$$\int_t^{+\infty} \Pi(\tau) U_{\tau} * Q_t d\tau = O(t^{-1}),$$

\noindent as $t \rightarrow +\infty$ in $L^2(\R^3)$.

\noindent \texttt{Step 3: Analysis of the first term in
\eqref{decomp-conv} in a particular case.}

\noindent Let us first show how to treat the terms with the phase
$\Theta (t)$ replaced by $\omega_{\infty} t$.

\noindent Note that
$$\widehat{Q}(x) = \frac{1}{\alpha +|x|^2},$$

\noindent Hence, in the case of the summands with $|z_{\infty}|^2$
it suffices to prove the integrability of the function
$$I(x) = \int_0^{\infty} e^{i(\omega_{\infty} \tau -\frac{|x|^2}{4\tau})}
\frac{\sqrt{\tau}}{(1 +\epsilon k_{\infty} \tau) (|x|^2 +4\alpha
\tau^2)} d\tau =$$
$$= A(x) \int_0^{\infty} e^{i(\omega_{\infty} \tau -\frac{|x|^2}{4\tau})} \left(
\frac{\sqrt{\tau}}{(1 +\epsilon k_{\infty} \tau)}
-\frac{4\alpha}{\epsilon k_{\infty}} \frac{\tau \sqrt{\tau}}{(|x|^2
+4\alpha \tau^2)} \right) d\tau +$$
$$+\frac{4\alpha}{\epsilon^2 k_{\infty}^2} A(x) \int_0^{\infty}
e^{i(\omega_{\infty} \tau -\frac{|x|^2}{4\tau})}
\frac{\sqrt{\tau}}{(|x|^2 +4\alpha \tau^2)} d\tau = I_1(x)
+I_2(x),$$

\noindent and the decay of
$$I_t(x) = \int_t^{\infty} e^{i(\omega_{\infty} \tau -\frac{|x|^2}{4\tau})}
\frac{\sqrt{\tau}}{(1 +\epsilon k_{\infty} \tau) (|x|^2 +4\alpha
\tau^2)} d\tau =$$
$$= A(x) \int_t^{\infty} e^{i(\omega_{\infty} \tau -\frac{|x|^2}{4\tau})} \left(
\frac{\sqrt{\tau}}{(1 +\epsilon k_{\infty} \tau)}
-\frac{4\alpha}{\epsilon k_{\infty}} \frac{\tau \sqrt{\tau}}{(|x|^2
+4\alpha \tau^2)} \right) d\tau +$$
$$+\frac{4\alpha}{\epsilon^2 k_{\infty}^2} A(x) \int_t^{\infty}
e^{i(\omega_{\infty} \tau -\frac{|x|^2}{4\tau})}
\frac{\sqrt{\tau}}{(|x|^2 +4\alpha \tau^2)} d\tau = I_{1,t}(x)
+I_{2,t}(x),$$

\noindent where $A(x) = \frac{\epsilon^2 k_{\infty}^2}{4\alpha
+\epsilon^2 k_{\infty}^2 |x|^2}.$

\noindent For the function $I_2(x)$ one has

$$|I_2(x)| \leq \frac{4\alpha}{\epsilon^2 k_{\infty}^2} A(x) \int_0^{\infty}
\frac{\sqrt{\tau}}{(|x|^2 +4\alpha \tau^2)} d\tau = C
\frac{A(x)}{\sqrt{|x|}} \in L^2(\R^3).$$

\noindent With the same estimate it is trivial to prove
$$I_{2,t}(x) = O(t^{-1/2})$$

\noindent as $t \rightarrow +\infty$, in $L^2(\R^3)$.

\noindent In order to treat $I_1$ note that
$$\frac{\sqrt{\tau}}{(1 +\epsilon k_{\infty} t)} -\frac{4\alpha}{\epsilon
k_{\infty}} \frac{\tau \sqrt{\tau}}{(|x|^2 +4\alpha \tau^2)} =
-\frac{1}{\epsilon k_{\infty} \sqrt{\tau} (1 +\epsilon k_{\infty}
\tau)} +\frac{|x|^2}{\epsilon k_{\infty} \sqrt{\tau} (|x|^2 +4\alpha
\tau^2)}.$$

\noindent Since $\frac{1}{\epsilon k_{\infty} \sqrt{\tau} (1
+\epsilon k_{\infty} \tau)} = O(t^{-3/2})$ as $t \rightarrow
+\infty$, is integrable on $(0, +\infty)$ and $A(x) \in L^2(\R^3)$
one has to prove
$$\frac{|x|^2 A(x)}{\epsilon k_{\infty}} \int_0^{+\infty} e^{i(\omega_{\infty}
\tau -\frac{|x|^2}{4\tau})} \frac{1}{\sqrt{\tau} (|x|^2 +4\alpha
\tau^2)} d\tau =$$
$$= A(x) \int_0^{+\infty} e^{i\omega_{\infty} (\tau -\frac{|x|^2}{4
\omega_{\infty} \tau})} \frac{1}{\sqrt{\tau}} d\tau -4 \alpha A(x)
\int_0^{+\infty} e^{i(\omega_{\infty} \tau -\frac{|x|^2}{4\tau})}
\frac{\tau^{3/2}}{(|x|^2 +4\alpha \tau^2)} d\tau \in L^2(\R^3).$$

\noindent From formulas 3.871.3 and 3.871.4 in \cite{tavole} one has
$$A(x) \int_0^{+\infty} e^{i\omega_{\infty} (\tau -\frac{|x|^2}{4
\omega_{\infty} \tau})} \frac{1}{\sqrt{\tau}} d\tau =
\frac{e^{i\pi/4}}{\sqrt{\pi \omega_{\infty}}} A(x) |x|^{3/2}
e^{-\sqrt{\omega_{\infty}} |x|} \in L^2(\R^3).$$

\noindent It remains to handle with the second integral in the
former sum which can be done integrating by parts in the following
way
$$\left| A(x) \int_0^{+\infty} e^{i(\omega_{\infty} \tau -\frac{|x|^2}{4\tau})}
\frac{\tau^{3/2}}{(|x|^2 +4\alpha \tau^2)} d\tau \right| =$$
$$= 4 A(x) \left| \int_0^{+\infty} e^{i (\omega_{\infty} \tau
-\frac{|x|^2}{4\tau})} \frac{d}{d\tau} \left[
\frac{\tau^{7/2}}{(|x|^2 +4\alpha \tau^2) (|x|^2 +4\omega_{\infty}
\tau^2)} \right] d\tau \right| \leq$$
$$\leq C A(x) \int_0^{+\infty} \frac{\tau^{5/2}}{(|x|^2 +4\min\{ \alpha,
\omega_{\infty} \} \tau^2)^2} d\tau \leq C \frac{A(x)}{\sqrt{|x|}}
\in L^2(\R^3).$$

\noindent Then we are done.

\noindent In order to estimate the decay of $I_{1,t}$ it suffices to
study the decay of
$$\frac{|x|^2 A(x)}{\epsilon k_{\infty}} \int_t^{+\infty} e^{i(\omega_{\infty}
\tau -\frac{|x|^2}{4\tau})} \frac{1}{\sqrt{\tau} (|x|^2 +4\alpha
\tau^2)} d\tau =$$
$$= A(x) \int_t^{+\infty} e^{i\omega_{\infty} (\tau -\frac{|x|^2}{4
\omega_{\infty} \tau})} \frac{1}{\sqrt{\tau}} d\tau -4 \alpha A(x)
\int_t^{+\infty} e^{i(\omega_{\infty} \tau -\frac{|x|^2}{4\tau})}
\frac{\tau^{3/2}}{(|x|^2 +4\alpha \tau^2)} d\tau,$$

\noindent which can be done integrating by parts as before. Let us
do that for the second term (the computation for the first one are
analogous and simpler):
$$\left|  A(x) \int_t^{+\infty} e^{i(\omega_{\infty} \tau -\frac{|x|^2}{4\tau})}
\frac{\tau^{3/2}}{(|x|^2 +4\alpha \tau^2)} d\tau \right| =$$
$$= 4  A(x) \left| \int_t^{+\infty} e^{i (\omega_{\infty} \tau
-\frac{|x|^2}{4\tau})} \frac{d}{d\tau} \left[
\frac{\tau^{7/2}}{(|x|^2 +4\alpha \tau^2) (|x|^2 +4\omega_{\infty}
\tau^2)} \right] d\tau \right| \leq$$
$$\leq C  A(x) \left[ t^{-1/2} +\int_t^{+\infty} \frac{\tau^{5/2}}{(|x|^2
+4\min\{ \alpha, \omega_{\infty} \} \tau^2)^2} d\tau \right] \leq$$
$$\leq C A(x) \left( 1+ \frac{1}{\sqrt{|x|}} \right) t^{-1/2}.$$

\noindent The case of the summands with $z_{\infty}^2$ is analogous,
while the case of $\overline{z_{\infty}}^2$ is more difficult
because $|x|^2 +4(\omega_{\infty} -2\xi_{\infty}) \tau^2 = 0$ for
$$\tau = t^* = \frac{|x|}{2 \sqrt{2\xi_{\infty} -\omega_{\infty}}}.$$

\noindent Let $g: \R^+ \rightarrow \R^+$ be a continuous function
with the properties:
$$0 < g(t^*) < t^* \quad \forall t^* > 0, \qquad \textrm{and} \qquad
A(x)g(t^*) \in L^2(\R^3).$$

\noindent It follows that $g(t^*) = O(t^*) = O(|x|)$ as $|x|
\rightarrow +\infty$. Hence, one can represent $(0, +\infty) = (0,
t^* -g(t^*)] \cup (t^* -g(t^*), t^* +g(t^*)] \cup (t^* +g(t^*),
+\infty)$. Integrating by parts once more one has
$$\left| A(x) \int_0^{t^* -g(t^*)} e^{i((\omega_{\infty} -2\xi_{\infty}) \tau
-\frac{|x|^2}{4\tau})} \frac{t^{3/2}}{|x|^2 +4\alpha \tau^2} d\tau
\right| \leq$$
$$\leq C A(x) \left( (t^* -g(t^*))^{-1/2} +\int_0^{t^* -g(t^*)}
\frac{t^{5/2}}{(|x|^2 +4\alpha \tau^2) ||x|^2 +4(\omega_{\infty}
-2\xi_{\infty}) \tau^2|} d\tau + \right.$$
$$\left. +\int_0^{t^* -g(t^*)} \frac{t^{9/2}}{(|x|^2 +4\alpha
\tau^2)^2 ||x|^2 +4(\omega_{\infty} -2\xi_{\infty}) \tau^2|} d\tau
+\right.$$
$$\left. +\int_0^{t^* -g(t^*)} \frac{t^{9/2}}{(|x|^2 +4\alpha
\tau^2) ||x|^2 +4(\omega_{\infty} -2\xi_{\infty}) \tau^2|^2} d\tau
\right) \leq$$
$$\leq C A(x) ((t^* -g(t^*))^{-1/2} +(t^* -g(t^*))^{3/8}) \in L^2(\R^3),$$

\noindent where the last inequality follows from formula 3.194.1 in
\cite{tavole}. In the same way (exploiting formula 3.194.2 instead
of 3.194.1 in \cite{tavole}), one has
$$\left| A(x) \int_{t^* +g(t^*)}^{\infty} e^{i((\omega_{\infty} -2\xi_{\infty})
\tau -\frac{|x|^2}{4\tau})} \frac{t^{3/2}}{|x|^2 +4\alpha \tau^2}
d\tau \right| \leq$$
$$\leq C A(x)((t^* +g(t^*))^{-1/8} +(t^* +g(t^*))^{-9/8} +(t^* -g(t^*))^{-1/2})
\in L^2(\R^3).$$

\noindent Finally,
$$\left| A(x) \int_{t^* -g(t^*)}^{t^* +g(t^*)} e^{i((\omega_{\infty}
-2\xi_{\infty}) \tau -\frac{|x|^2}{4\tau})} \frac{t^{3/2}}{|x|^2
+4\alpha \tau^2} d\tau \right| \leq$$
$$\leq C A(x) \int_{t^* -g(t^*)}^{t^* +g(t^*)} \frac{1}{\sqrt{\tau}} d\tau
\leq C \frac{A(x) g(t^*)}{\sqrt{t^* -g(t^*)}} \in L^2(\R^3).$$

\noindent Summing up, the integrability of the integral function
$$\int_0^{+\infty} \Pi(\tau) U_{\tau} * Q d\tau$$

\noindent is achieved. It is left to study the decay of
$$A(x) \int_t^{+\infty} e^{i((\omega_{\infty} -2\xi_{\infty})
\tau -\frac{|x|^2}{4\tau})} \frac{t^{3/2}}{|x|^2 +4\alpha \tau^2}
d\tau.$$

\noindent First of all, let us note that integrating by parts one
obtains
$$\left| A(x) \int_{(0, t^* -g(t^*)] \cap [t, +\infty)} e^{i((\omega_{\infty}
-2\xi_{\infty}) \tau -\frac{|x|^2}{4\tau})} \frac{t^{3/2}}{|x|^2
+4\alpha \tau^2} d\tau \right| \leq$$
$$\leq C A(x) \int_t^{t^* -g(t^*)} \left| \frac{d}{d\tau} \frac{t^{7/2}}{(|x|^2
+4\alpha \tau^2) (|x|^2 +4(\omega_{\infty} -2\xi_{\infty}) \tau^2)}
\right| d\tau \leq$$
$$\leq C A(x) \left( t^{-1/2} +\int_t^{t^* -g(t^*)} \frac{\sqrt{\tau}}{|x|^2
+4\alpha \tau^2} d\tau +\int_t^{t^* -g(t^*)}
\frac{\sqrt{\tau}}{||x|^2 +4(\omega_{\infty} -2\xi_{\infty})
\tau^2|} d\tau + \right.$$
$$\left. +\int_t^{t^* -g(t^*)} \frac{\tau^{9/2}}{(|x|^2 +4\alpha
\tau^2) ||x|^2 +4(\omega_{\infty} -2\xi_{\infty}) \tau^2|^2} d\tau
\right).$$

\noindent The three integrals in the last inequality can be
estimated in the following way:
\begin{itemize}
\item[(i)]  $\int_t^{t^* -g(t^*)} \frac{\sqrt{\tau}}{|x|^2
+4\alpha \tau^2} d\tau \leq C \int_t^{+\infty} \tau^{-3/2} d\tau
\leq C t^{-1/2}$;
\item[(ii)] $\int_t^{t^* -g(t^*)} \frac{\sqrt{\tau}}{|x|^2
+4(2\xi_{\infty} -\omega_{\infty}) \tau^2} d\tau = \int_t^{t^*
-g(t^*)} \frac{\sqrt{\tau}}{(|x| +2\sqrt{2\xi_{\infty}
-\omega_{\infty}} \tau) ||x| -2\sqrt{2\xi_{\infty} -\omega_{\infty}}
\tau|} d\tau$\\ $\leq C t^{-1/2} \int_0^{t^* -g(t^*)} \frac{1}{||x|
-2\sqrt{2\xi_{\infty} -\omega_{\infty}} \tau|} d\tau \leq C
t^{-1/2}$;
\item[(iii)]    $\int_t^{t^* -g(t^*)} \frac{\tau^{9/2}}{(|x|^2 +4\alpha
\tau^2) ||x|^2 +4(\omega_{\infty} -2\xi_{\infty}) \tau^2|^2} d\tau
\leq C t^{-1/2} \int_0^{t^* -g(t^*)} \frac{\tau}{||x|
-2\sqrt{\omega_{\infty} -2\xi_{\infty}} \tau|^2} d\tau$\\ $\leq C
t^{-1/2} (1 +\ln ||x| -2\sqrt{\omega_{\infty} -2\xi_{\infty}} (t^*
-g(t^*))|)$.
\end{itemize}

\noindent Hence, since $A(x) \ln ||x| -2\sqrt{\omega_{\infty}
-2\xi_{\infty}} (t^* -g(t^*))| \in L^2(\R^3)$, one can conclude
$$A(x) \int_{(0, t^* -g(t^*)] \cap [t, +\infty)} e^{i((\omega_{\infty}
-2\xi_{\infty}) \tau -\frac{|x|^2}{4\tau})} \frac{t^{3/2}}{|x|^2
+4\alpha \tau^2} d\tau = O(t^{-1/2})$$

\noindent as $t \rightarrow +\infty$, in $L^2(\R^3)$.

\noindent Let us now observe that
$$\left| A(x) \int_{(t^* -g(t^*), +\infty) \cap [t, +\infty)}
e^{i((\omega_{\infty} -2\xi_{\infty}) \tau -\frac{|x|^2}{4\tau})}
\frac{t^{3/2}}{|x|^2 +4\alpha \tau^2} d\tau \right| \leq$$
$$\leq C A(x) \int_t^{+\infty} \left| \frac{d}{d\tau} \frac{t^{7/2}}{(|x|^2
+4\alpha \tau^2) (|x|^2 +4(\omega_{\infty} -2\xi_{\infty}) \tau^2)}
\right| d\tau \leq $$
$$\leq B(x) A(x) \left( t^{-1/2} +\int_t^{+\infty} \frac{\sqrt{\tau}}{|x|^2
+4\alpha \tau^2} d\tau \right) \leq C B(x) A(x) t^{-1/2} \in
L^2(\R^3),$$

\noindent where $B:\R^3 \rightarrow \R^+$ is a continuous bounded
function.

\noindent Finally,
$$\left| A(x) \int_{(t^* -g(t^*), (t^* +g(t^*)] \cap [t, +\infty)}
e^{i((\omega_{\infty} -2\xi_{\infty}) \tau -\frac{|x|^2}{4\tau})}
\frac{t^{3/2}}{|x|^2 +4\alpha \tau^2} d\tau \right| \leq$$
$$\leq C A(x) \int_t^{t^* +g(t^*)} \frac{1}{\sqrt{\tau}} d\tau\leq C A(x)
g(t^*) t^{-1/2} \in L^2(\R^3).$$

\noindent Summing up, thanks to the unitarity of $U_t$, we proved
$$U_t * \int_t^{+\infty} \Pi(\tau) U_{\tau} * Q d\tau = O(t^{-1/2}),$$

\noindent as $t \rightarrow +\infty$, in $L^2(\R^3)$.

\noindent \texttt{Step 4: conclusion of the proof.}

\noindent The conclusions of the previous step hold true if the
phase $\omega_{\infty} t$ is replaced by $\Theta(t)$. In fact, the
estimates which involve the integral of the absolute value are
totally unaffected by change of phase, then it is only left to
adjust the argument involving integration by parts. This can be done
integrating by parts exactly as before, which leaves a factor
$e^{i(\Theta(t) -\omega_{\infty} t)}$ in the integrand. Then, the
boundary terms can be treated in the same way because
$|e^{i(\Theta(t) -\omega_{\infty} t)}| = 1$. Finally, the extra
contribution to the integrand can be estimated as it is done for the
summand arising from differentiation of $t^{7/2}$ since
$|\dot{\Theta}(t) -\omega_{\infty}| \leq \frac{C}{1 +\epsilon
k_{\infty} t}$ for all $t > 0$, where $C$ is a positive constant.
\end{proof}

\noindent Summing up, we have proved the following asymptotic
stability result.

\begin{theo}
Let $\sigma \in \left( \frac{1}{\sqrt{2}}, \sigma^* \right)$, for a certain
$\sigma^* \in
\left( \frac{1}{\sqrt{2}}, \frac{\sqrt{3} +1}{2 \sqrt{2}} \right]$ and $u(t) \in
C(\R^+, V)$ be a solution of equation \eqref{eq1} with
$$u(0) = u_0 = e^{i\omega_0 t +\gamma_0} \Phi_{\omega_0} +e^{i\omega_0 t
+\gamma_0} [(z_0 +\overline{z_0}) \Psi_1 +i (z_0 -\overline{z_0})
\Psi_2] +f_0 \in V \cap L^1_w(\R^3),$$

\noindent for some $\omega_0 > 0$, $\gamma_0$, $z_0 \in \R$ and $f_0
\in L^2(\R^3) \cap L^1_w(\R^3)$. Furthermore, assume that the
initial datum $u_0$ is close to a solitary wave, i.e.
$$|z_0| \leq \epsilon^{1/2} \qquad \textrm{and} \qquad \|f_0\|_{L^1_w} \leq
c \epsilon^{3/2},$$

\noindent where $c$, $\epsilon > 0$.

\noindent Then, if $\epsilon$ is sufficiently small, the solution
$u(t)$ can be asymptotically decomposed as follows
$$u(t) = e^{i\omega_{\infty} t +i b_1 \log (1 +\epsilon k_{\infty} t)}
\Phi_{\omega_{\infty}} +U_t*\phi_{\infty} +r_{\infty}(t), \quad
\textrm{as} \;\; t \rightarrow +\infty,$$

\noindent where $\omega_{\infty}$, $\epsilon k_{\infty} > 0$, $b_1
\in \R$ and $\phi_{\infty}$, $r_{\infty}(t) \in L^2(\R^3)$ with
$$\| r_{\infty}(t) \|_{L^2} = O(t^{-1/4}) \quad \textrm{as} \;\; t \rightarrow
+\infty,$$

\noindent in $L^2(\R^3)$.
\end{theo}

\begin{rmk}
Numerical evidences (see Lemma \ref{reale}) suggest $\sigma^* =
\frac{\sqrt{3} +1}{2 \sqrt{2}} \simeq 0,96$.
\end{rmk}

\section{Appendices}
In the following appendices we collect auxiliary material on the model
studied. In Appendices A and B we recall results from \cite{ADO}
regarding resolvent, spectrum and dispersive behavior of linearization
operator $L$ (see \eqref{linearizzato}) to help the independent
reading of the present paper. In the subsequent appendices C,D,E we
state and prove further details regarding spectral properties of the
linearized operators, in particular the structure of eigenfunctions
associated to the discrete spectrum, the structure of the generalized
eignevectors, and the proof of Lemma II.8. 
\subsection{Resolvent and spectrum of linearization} 
We denote
\be \label{gom}
G_{\omega \pm i\lambda}(x) = \frac{e^{i \sqrt{-\omega \mp
i\lambda} |x|}}{4\pi |x|} \qquad \omega > 0, \lambda \in \C,
\ee

\noindent with the prescription $\Im{\sqrt{-\omega \pm i\lambda}} >
0$.

\noindent Furthermore, we make use of the notation $\langle g,h \rangle : =
\int_{\R^3} g(x) h(x) \, dx$.

\noindent
The resolvent of linearized operator $L$
is described in the following 
\begin{theo}    \label{ris}
The resolvent $R(\lambda) = (L -\lambda I)^{-1}$ of the operator $L$
is given by
\begin{equation}    \label{eq:risolvente}
R(\lambda) = \left[
  \begin{array}{cc}
    -\lambda \mathcal{G}_{\lambda^2} * & -\Gamma_{\lambda^2} * \\
    \Gamma_{\lambda^2} * & -\lambda \mathcal{G}_{\lambda^2} * \\
  \end{array}
\right] +\frac{4\pi}{W(\lambda^2)} i \left[
  \begin{array}{cc}
    \Lambda_1 & i\Sigma_2 \\
    -i\Sigma_1 & \Lambda_2\\
  \end{array}
\right],
\end{equation}
where
$$W (\lambda^2) = 32\pi2 \alpha_1 \alpha_2 -4i \pi (\alpha_1 +\alpha_2)
\left(\sqrt{-\omega +i\lambda} +\sqrt{-\omega -i \lambda} \right)
-2\sqrt{-\omega +i\lambda} \sqrt{-\omega -i\lambda},$$
and formula \eqref{eq:risolvente} holds
for all $\lambda \in \C \setminus \{ \lambda \in \C: \;
W(\lambda^2) = 0, \;\; \textrm{or} \;\; \Re(\lambda) = 0 \;
\textrm{and} \; |\Im(\lambda)| \geq \omega \}$. Furthermore, the
symbol $*$ in  \eqref{eq:risolvente}
denotes the convolution and
$$
\mathcal{G}_{\lambda^2}(x) = \frac{1}{2i \lambda} \left(
G_{\omega -i\lambda}(x) -G_{\omega +i\lambda}(x) \right), \quad 
\Gamma_{\lambda^2}(x) = \frac{1}{2} \left( G_{\omega
-i\lambda}(x) +G_{\omega +i\lambda}(x) \right).
$$
Finally, the entries of the second matrix are finite rank operators
whose action on $f \in L^2 (\R^3)$ reads
\begin{equation} \label{lambdasigma}
\Lambda_1 f = [i\lambda (4\pi \alpha_2 -i\sqrt{-\omega +i\lambda})
\langle\mathcal{G}_{\lambda^2}, f \rangle -(4\pi \alpha_1 -i\sqrt{-\omega
+i\lambda}) \langle\Gamma_{\lambda^2}, f \rangle]G_{\omega +i\lambda}
+
\end{equation}
$$+[i\lambda (4\pi \alpha_2 -i\sqrt{-\omega -i\lambda})
\langle \mathcal{G}_{\lambda^2}, f \rangle +(4\pi \alpha_1 -i\sqrt{-\omega
-i\lambda}) \langle \Gamma_{\lambda^2}, f \rangle]G_{\omega -i\lambda},$$
$$\Lambda_2 f = [i\lambda (4\pi \alpha_1 -i\sqrt{-\omega +i\lambda})
\langle \mathcal{G}_{\lambda^2}, f \rangle -(4\pi \alpha_2 -i\sqrt{-\omega
+i\lambda}) \langle \Gamma_{\lambda^2}, f \rangle]G_{\omega +i\lambda} +$$
$$+[i\lambda (4\pi \alpha_1 -i\sqrt{-\omega -i\lambda})
\langle \mathcal{G}_{\lambda^2}, f \rangle +(4\pi \alpha_2 -i\sqrt{-\omega
-i\lambda}) \langle \Gamma_{\lambda^2}, f\rangle]G_{\omega -i\lambda},$$
$$\Sigma_1 f = -[i\lambda (4\pi \alpha_2 -i\sqrt{-\omega +i\lambda})
\langle \mathcal{G}_{\lambda^2}, f \rangle -(4\pi \alpha_1 -i\sqrt{-\omega
+i\lambda}) \langle \Gamma_{\lambda^2}, f \rangle]G_{\omega +i\lambda} +$$
$$+[i\lambda (4\pi \alpha_2 -i\sqrt{-\omega -i\lambda})
\langle \mathcal{G}_{\lambda^2}, f \rangle +(4\pi \alpha_1 -i\sqrt{-\omega
-i\lambda}) \langle \Gamma_{\lambda^2}, f \rangle]G_{\omega -i\lambda},$$
$$\Sigma_2 f = -[i\lambda (4\pi \alpha_1 -i\sqrt{-\omega +i\lambda})
\langle \mathcal{G}_{\lambda^2}, f\rangle -(4\pi \alpha_2 -i\sqrt{-\omega
+i\lambda}) \langle \Gamma_{\lambda^2}, f\rangle ]G_{\omega +i\lambda} +$$
$$+[i\lambda (4\pi \alpha_1 -i\sqrt{-\omega -i\lambda})
\langle \mathcal{G}_{\lambda^2}, f \rangle +(4\pi \alpha_2 -i\sqrt{-\omega
-i\lambda}) \langle \Gamma_{\lambda^2}, f \rangle]G_{\omega -i\lambda}.$$
\end{theo}
\vskip10pt
The resolvent determines also the spectrum of the operator $L$, as given in the following

\begin{prop} The spectrum of the linearized operator L has the following structure
\begin{itemize}
 \item [(a)] $\sigma_{ess}(L) =\{ \lambda \in \C: \; \Re(\lambda) = 0 \,
\text{and} \, |\Im(\lambda)| \geq \omega \}$
 \item[(b)] If $\sigma \in (0, 1/\sqrt{2})$, the only eigenvalue of $L$ is $0$ with algebraic multiplicity $2$.
 \item[(c)] If  $\sigma =1/\sqrt{2}$,  L has resonances $\pm i\omega$ at the border of the essential spectrum and the eigenvalue $0$ with algebraic multiplicity $2$.
 \item[(d)] If $\sigma \in (1/\sqrt{2}, 1)$, $L$ has two simple eigenvalues $\pm i\xi=
            =\pm i 2\sigma \sqrt{1 -\sigma^2} \omega$ and the eigenvalue $0$ with algebraic multiplicity $2$.
 \item[(e)] If $\sigma = 1$, the only eigenvalue of $L$ is $0$ with algebraic multiplicity $4$.
 \item[(f)] If $\sigma \in (1, +\infty)$, $L$ has two simple eigenvalues $\pm 2\sigma \sqrt{\sigma^2 -1} \omega$
            and the eigenvalue $0$ with algebraic multiplicity $2$.
\end{itemize}

\end{prop}


\subsection{Dispersive estimates}
We recall that the propagator $e^{-t L}$ is the inverse Laplace transform of the resolvent. So the dispersive behaviour associated to the linearized dynamics projected on the continuous spectrum is controlled by the following result, proved in \cite{ADO}, Theorem 4.8.
\begin{theo}    \label{stima-cont}
Let $\sigma\neq \frac{1}{{\sqrt 2}}$. There exists a constant $C > 0$ such that
$$\left| \frac{1}{2\pi i} \int_{\R^3} \int_{\mathcal{C}_+ \cup \mathcal{C}_-}
(R(\lambda +0) -R(\lambda -0)) (x) e^{-\lambda t} f(y) \, d\lambda  dy
\right| \leq C \left( 1 +\frac{1}{|x|} \right) t^{-\frac{3}{2}}
\int_{\R^3} \left( 1 +\frac{1}{|y|} \right) |f(y)| dy$$

\noindent for any $f \in L^1_{w}(\R^3)$, where
$$
\mathcal{C}_+  = \{ \lambda \in \C: \;
\Re(\lambda) =
0 \, \textrm{and} \, \ \Im(\lambda) \geq \omega \},\ \ \   \mathcal{C}_- =
\{ \lambda \in \C: \; \Re(\lambda) = 0 \, \textrm{and} \, \
\Im(\lambda) \leq -\omega \}\ .
$$
\end{theo}

\subsection{Eigenfunctions associated to $\pm i\xi$ and generalized
eigenfunctions}
\subsubsection{The eigenfunctions associated to $\pm i\xi$  \label{autofunz}}
Here we describe the eigenspaces associated to the simple
purely imaginary eigenvalues $\pm i\xi = \pm i 2\sigma \sqrt{1
-\sigma^2} \omega$.

\noindent Let us start with the eigenvalue $i\xi$. The following
proposition holds true.
\begin{prop}
The eigenspace associated to $i\xi$ is spanned by
$$\Psi(x) = \left( \begin{array}{cc}
                   \Psi_1(x)\\
                   \Psi_2(x)
                  \end{array}
\right) = \frac{e^{-\sqrt{\omega -\xi} |x|}}{4\pi |x|} \left(
\begin{array}{cc}
       1\\
       i
       \end{array}
\right) -\frac{\sqrt{1 -\sigma^2} -1}{\sigma} \,
\frac{e^{-\sqrt{\omega +\xi} |x|}}{4\pi |x|} \left(
\begin{array}{cc}
       1\\
       -i
       \end{array}
\right).$$

\end{prop}

\begin{proof}
In order to prove the proposition we need to solve the equation
$$L \Psi = i\xi \Psi$$

\noindent in $D(L)$. For $x \neq 0$, the previous equation is
equivalent to the system
$$\left\{ \begin{array}{ll}
     (-\triangle +\omega)^2 \Psi_1 -\xi^2 \Psi_1 = 0\\
     \Psi_2 = \frac{i}{\xi} (-\triangle +\omega) \Psi_1
         \end{array}
\right.,$$

\noindent from which follows that $\Psi_1$ must belong to
$L^2(\R^3)$ and solve the equation
$$(-\triangle +\omega -\xi) (-\triangle +\omega +\xi) \Psi_1 = 0.$$

\noindent Hence, the solutions in $L^2(\R^3)$ are of the form
$$\left\{ \begin{array}{ll}
     \Psi_1(x) = A \frac{e^{-\sqrt{\omega -\xi} |x|}}{4\pi |x|} +B
\frac{e^{-\sqrt{\omega +\xi} |x|}}{4\pi |x|}\\
     \Psi_2(x) = iA \frac{e^{-\sqrt{\omega -\xi} |x|}}{4\pi |x|} -iB
\frac{e^{-\sqrt{\omega +\xi} |x|}}{4\pi |x|}
         \end{array}
\right.,$$

\noindent for any $A$, $B \in \C$.

\noindent It is left to look for $A$, $B \in \C$ such that $\Psi_i
\in D(L_i)$ for $i = 1$, $2$, i.e.
$$\left\{ \begin{array}{ll}
     -\frac{\sqrt{\omega -\xi}}{4\pi} A -\frac{\sqrt{\omega +\xi}}{4\pi} B
= -(2\sigma +1) \frac{\sqrt{\omega}}{4\pi} (A +B)\\
     -i\frac{\sqrt{\omega -\xi}}{4\pi} A +i\frac{\sqrt{\omega +\xi}}{4\pi} B
= -\frac{\sqrt{\omega}}{4\pi} (iA -iB)
         \end{array}
\right..$$

\noindent Exploiting the fact that $\xi = 2\sigma \sqrt{1 -\sigma^2}
\omega$ one can show that the two equations of the previous system
are linearly dependent and
$$B = -\frac{\sqrt{1 -\sigma^2} +1}{\sigma} A.$$

\noindent The thesis follows by setting $A = 1$.
\end{proof}

\noindent Let us note that in the previous proof we have chosen the
constant in such a way that $\Psi_1(x) \in \R$ and $\Psi_2(x) \in
i\R$ for any $x \in \R^3 \setminus \{0\}$. This fact will be used to
prove the next proposition.
\begin{prop}
The eigenspace associated to $-i\xi$ is spanned by
$$\Psi^* = \left( \begin{array}{cc}
                   \Psi_1\\
                   -\Psi_2
                  \end{array}
\right).$$
\end{prop}

\begin{proof}
In the previous proposition we proved that
$$\left\{ \begin{array}{ll}
     L_2 \Psi_2 = i\xi \Psi_1\\
     -L_1 \Psi_1 = i\xi \Psi_2
         \end{array}
\right.,$$

\noindent with $\Psi_1$ real and $\Psi_2$ purely imaginary.

\noindent Taking the conjugate of both equations and recalling that
the operators $L_i$, $i = 1$, $2$ act on the real and imaginary
parts separately, one has
$$\left\{ \begin{array}{ll}
     L_2 (-\Psi_2) = -i\xi \Psi_1\\
     -L_1 \Psi_1 = -i\xi (-\Psi_2)
         \end{array}
\right.,$$

\noindent which is equivalent to
$$L \Psi^* = -i\xi \Psi^*,$$

\noindent because the operators $L_i$, $i = 1$, $2$ are linear. The
proof is complete.
\end{proof}

\subsubsection{The generalized eigenfunctions}    \label{gen-autofunz}
Our goal is to compute the generalized eigenfunctions associated to
the continuous spectrum. In order to do that, we treat the two
branches $\mathcal{C}_+$ and $\mathcal{C}_-$ of the continuous
spectrum separately.

\begin{prop}
The generalized eigenfunctions associated to $\mathcal{C}_+$ are
$$\Psi_+ (x) = A \frac{e^{-\sqrt{\omega +\eta} |x|}}{4\pi |x|} \left(
\begin{array}{cc}
       1\\
       -i
       \end{array}
\right) +C \frac{e^{-i \sqrt{\eta -\omega} |x|}}{4\pi |x|} \left(
\begin{array}{cc}
       1\\
       i
       \end{array}
\right) +D \frac{e^{i \sqrt{\eta -\omega} |x|}}{4\pi |x|} \left(
\begin{array}{cc}
       1\\
       i
       \end{array}
\right),$$

\noindent  for any $\eta \in [\omega, +\infty)$ and $D \in \C$, with
$$\begin{array}{ll}
A = \frac{\sigma \sqrt{\omega}}{\sqrt{\omega +\eta} -(\sigma +1)
\sqrt{\omega}} (C +D),\\
C = \frac{(2\sigma +1) \omega +(\sigma +1) \sqrt{\omega}
(i\sqrt{\eta -\omega} -\sqrt{\eta +\omega}) -i\sqrt{\eta^2
-\omega^2}}{-(2\sigma +1) \omega +(\sigma +1) \sqrt{\omega}
(i\sqrt{\eta -\omega} +\sqrt{\eta +\omega}) -i\sqrt{\eta^2
-\omega^2}} D.\\
\end{array}$$

\end{prop}

\begin{proof}
For any $\eta \in [\omega, +\infty)$, we need to solve the system
$$L\Psi_+ = i\eta \Psi_+,$$

\noindent where $\Psi_+ \in L^{\infty}(\R^3)$ does not necessary
belongs to $L^2(\R^3)$. As in the computation for the eigenfunction
at $\pm i \xi$, if $x \neq 0$ the former equation is equivalent to
the system
$$\left\{ \begin{array}{ll}
     (-\triangle +\omega -\xi) (-\triangle +\omega +\xi) (\Psi_{+})_1 = 0\\
     (\Psi_{+})_2 = \frac{i}{\xi} (-\triangle +\omega) (\Psi_{+})_1
         \end{array}
\right.,$$

\noindent which leads to
$$\Psi_+ (x) = A \frac{e^{-\sqrt{\omega +\eta} |x|}}{4\pi |x|} \left(
\begin{array}{cc}
       1\\
       -i
       \end{array}
\right) +B \frac{e^{\sqrt{\omega +\eta} |x|}}{4\pi |x|} \left(
\begin{array}{cc}
       1\\
       i
       \end{array}
\right) +C \frac{e^{-i \sqrt{\eta -\omega} |x|}}{4\pi |x|} \left(
\begin{array}{cc}
       1\\
       i
       \end{array}
\right) +D \frac{e^{i \sqrt{\eta -\omega} |x|}}{4\pi |x|} \left(
\begin{array}{cc}
       1\\
       i
       \end{array}
\right),$$

\noindent for some $A$, $B$, $C$, $D \in \C$. Since we require
$\Psi_+ \in L^{\infty}(\R^3)$, we get $B = 0$. Moreover, the
boundary conditions in the domain of the operators $L_1$ and $L_2$
must be satisfied by $(\Psi_{+})_1$ and $(\Psi_{+})_2$ respectively.
Then $A$, $C$, and $D$ solve the system
$$\left\{ \begin{array}{ll}
       -\frac{\sqrt{\omega +\eta}}{4\pi} A -i\frac{\sqrt{\eta +\omega}}{4\pi} C +i\frac{\sqrt{\eta +\omega}}{4\pi} D
       = -\frac{(2\sigma +1) \sqrt{\omega}}{4\pi} (A +C +D)\\
       i\frac{\sqrt{\omega +\eta}}{4\pi} A +\frac{\sqrt{\eta +\omega}}{4\pi} C -\frac{\sqrt{\eta +\omega}}{4\pi} D
       = -\frac{\sqrt{\omega}}{4\pi} (-i A +i C +i D)
\end{array} \right.,$$

\noindent which concludes the proof.
\end{proof}

\noindent In the same way, one can prove the analogous result about
$\mathcal{C}_-$.
\begin{prop}
The generalized eigenfunctions associated to $\mathcal{C}_-$ are
$$\Psi_- (x) = A \frac{e^{-\sqrt{\omega -\eta} |x|}}{4\pi |x|} \left(
\begin{array}{cc}
       1\\
       i
       \end{array}
\right) +C \frac{e^{-i \sqrt{-(\eta +\omega)} |x|}}{4\pi |x|} \left(
\begin{array}{cc}
       1\\
       -i
       \end{array}
\right) +D \frac{e^{i \sqrt{-(\eta +\omega)} |x|}}{4\pi |x|} \left(
\begin{array}{cc}
       1\\
       -i
       \end{array}
\right),$$

\noindent  for any $\eta \in (-\infty, -\omega]$, where $D \in \C$
and
$$\begin{array}{ll}
A = \frac{\sigma \sqrt{\omega}}{\sqrt{\omega -\eta} -(\sigma +1)
\sqrt{\omega}} (C +D),\\
C = \frac{(2\sigma +1) \omega +(\sigma +1) \sqrt{\omega}
(i\sqrt{-(\eta +\omega)} -\sqrt{\omega -\eta}) -i\sqrt{\eta^2
-\omega^2}}{-(2\sigma +1) \omega +(\sigma +1) \sqrt{\omega}
(i\sqrt{-(\eta +\omega)} +\sqrt{\eta +\omega}) -i\sqrt{\eta^2
-\omega^2}} D.\\
\end{array}$$

\end{prop}

It is easy to see that the projection operators  from $L^2(\R^3)$ onto $X^0$,
$X^1$ and $X^c$ are given by
$$\begin{array}{ll}
\displaystyle P^0 f = -\frac{2}{\Delta} \Omega \left( f,
\frac{d\Phi_\omega}{d\omega} \right) J\Phi_{\omega}
+\frac{2}{\Delta} \Omega \left( f, J\Phi_{\omega} \right)
\frac{d\Phi_{\omega}}{d\omega}, & \Delta = \frac{d}{d\omega} \|
\Phi_{\omega}
\|_{L^2},\\
\displaystyle P^1 f = \frac{\Omega(f,\Psi)}{\kappa} \Psi
+\frac{\Omega(f,\Psi^*)}{\kappa} \Psi^*, & \kappa = \Omega(\Psi, \Psi^*),\\
\displaystyle P^c f = f -P^0 f -P^1 f,
\end{array}$$

Moreover, we denote with $\Pi^{\pm}$ the
projections onto the branches $\mathcal{C}_{\pm}$ of the continuous
spectrum separately.

\subsection{Proof of Lemma \ref{lemma-LM}} \label{dim-lemma}
In this appendix we prove Lemma \ref{lemma-LM} whose statement is
recalled for the reader's convenience.
\begin{lemma}
There exists a constant $C > 0$ such that for each $h \in X^c$ holds
$$\left\| [P^c J -i(\Pi^+ -\Pi^-)]h  \right\|_{L^1_w} \leq C \| h
\|_{L^{\infty}_{w^{-1}}}.$$

\end{lemma}

\begin{proof}
From the definitions of the operators $P^c$ and $\Pi^{\pm}$ one gets
$$P^c J -i(\Pi^+ -\Pi^-) = \Pi^+ (J -iI) +\Pi^- (J +iI) =$$
$$= \frac{1}{2\pi i} \left[ \int_{\mathcal{C}^+} (R(\lambda +0) -R(\lambda -0))
(J -iI) d\lambda +\int_{\mathcal{C}^-} (R(\lambda +0) -R(\lambda
-0)) (J +iI) d\lambda \right].$$

\noindent We will estimate just the first integral because the
second one can be handled in the same way. Exploiting the explicit
form of the resolvent \eqref{ris} it follows that
$$R(\lambda) (J -iI) = \left( i\lambda \mathcal{G}_{\lambda^2}*
+\Gamma_{\lambda^2}* \right) \left[
  \begin{array}{cc}
    1 & i \\
    -i & 1 \\
  \end{array}
\right] +\frac{4\pi}{D(\lambda^2)} \left[
  \begin{array}{cc}
    \Lambda_1 +\Sigma_2 & i(\Lambda_1 +\Sigma_2) \\
    -i(\Lambda_2 +\Sigma_1) & \Lambda_2 +\Sigma_1 \\
  \end{array}
\right] =$$
$$= R_*(\lambda) (J -iI) +R_m(\lambda) (J -iI),$$

\noindent where $R_*$ and $R_m$ correspond to the convolution term
of the resolvent and the multiplicative term.

\noindent Note that
$$i\lambda \mathcal{G}_{\lambda^2}(x-y) +\Gamma_{\lambda^2}(x-y) = 2 G_{\omega
-i\lambda}(x-y) = \frac{e^{i\sqrt{-\omega +i\lambda} |x-y|}}{2\pi
|x-y|}$$

\noindent is continuous on $\mathcal{C}_+$. Hence, the integral on
$\mathcal{C}_+$ of the convolution addends vanishes.

\noindent Let us now consider the multiplicative addends in the
integral on $\mathcal{C}_+$. From the explicit formulas for
$\Lambda_1$ and $\Sigma_2$ given in Proposition \ref{ris} one can
compute
$$(\Lambda_1 +\Sigma_2)(x,y) = 8\pi (\alpha_2 -\alpha_1) G_{\omega
-i\lambda}(y) G_{\omega +i\lambda}(x) +[8\pi (\alpha_2 +\alpha_1)
-4i \sqrt{-\omega -i\lambda}] G_{\omega -i\lambda}(y) G_{\omega
-i\lambda}(x) =$$
$$= 4\sigma \sqrt{\omega} \frac{e^{i\sqrt{-\omega +i\lambda} |y|}}{4\pi
|y|} \frac{e^{i\sqrt{-\omega -i\lambda} |x|}}{4\pi |x|} -[4(\sigma
+1) \sqrt{\omega} -4i \sqrt{-\omega -i\lambda}]
\frac{e^{i\sqrt{-\omega +i\lambda} (|x| +|y|)}}{(4\pi)^2 |x||y|}.$$

\noindent Denote
$$D_{\pm}(\lambda^2) = D((\lambda \pm 0)^2).$$

\noindent Then it follows
$$\int_{\mathcal{C}^+} [(R_m(\lambda +0) -R_m(\lambda -0))
(J -iI)]_{1,1} d\lambda =$$
$$=\int_{\mathcal{C}^+} \frac{\sigma \sqrt{\omega} e^{i \sqrt{-\omega +i\lambda}
|y|} e^{-i \sqrt{-\omega -i\lambda} |x|} +((\sigma +1) \sqrt{\omega}
-i\sqrt{-\omega -i\lambda}) e^{i\sqrt{-\omega +i\lambda} (|x|
+|y|)}}{\pi |x| |y| D_+(\lambda^2)} d\lambda +$$
$$-\int_{\mathcal{C}^+} \frac{\sigma \sqrt{\omega} e^{i \sqrt{-\omega +i\lambda}
|y|} e^{i \sqrt{-\omega -i\lambda} |x|} +((\sigma +1) \sqrt{\omega}
+i\sqrt{-\omega -i\lambda}) e^{i\sqrt{-\omega +i\lambda} (|x|
+|y|)}}{\pi |x| |y| D_-(\lambda^2)} d\lambda.$$

\noindent If we compute the change of variable $k = \sqrt{-\omega
-i\lambda}$ in the first integral of the last equality, and $k =
-\sqrt{-\omega -i\lambda}$ in the second one, then one has
$$\left| \int_{\mathcal{C}^+} [(R_m(\lambda +0) -R_m(\lambda -0))
(J -iI)]_{1,1} d\lambda \right| =$$
$$=\left| \frac{4i}{\pi |x| |y|} \left( \int_{-\infty}^{+\infty} \sigma
\sqrt{\omega} k \frac{e^{-\sqrt{k^2 +2\omega} |y|}}{D(k)} e^{-ik|x|}
dk +\int_{-\infty}^{+\infty} 2ik^2 \frac{e^{-\sqrt{k^2 +2\omega}
(|x| +|y|)}}{D(k)} dk \right) \right| \leq$$
$$\leq C \frac{e^{-\sqrt{2\omega}|y|}}{|y| |x|} \min \left\{ \frac{1}{|x|},
e^{-\sqrt{2\omega} |x|} \right\} \leq C
\frac{e^{-\sqrt{2\omega}|y|}}{|y|}
\frac{e^{-\sqrt{2\omega}|x|}}{|x|},$$

\noindent where the first inequality is obtained integrating by
parts both integrals.

\noindent The integral of the other three elements of the matrix 
operator $(R_m(\lambda +0) -R_m(\lambda -0)) (J -iI)$ can be
estimated in the same way and this implies the statement of the
lemma.
\end{proof}



\end{document}